\documentclass[12pt,a4paper]{article} 

\usepackage[utf8]{inputenc}
\usepackage[english]{babel}
\usepackage[T1]{fontenc}
\usepackage{dirtytalk} 
\usepackage{ragged2e}
\usepackage{geometry}
\usepackage[table]{xcolor}
\definecolor{olive}{rgb}{0.3, 0.4, .1}
\definecolor{fore}{RGB}{249,242,215}
\definecolor{back}{RGB}{51,51,51}
\definecolor{title}{RGB}{255,0,90}
\definecolor{blackViolet}{RGB}{138,43,226}
\definecolor{gold}{rgb}{1.,0.84,0.}
\definecolor{JungleGreen}{cmyk}{0.99,0,0.52,0}
\definecolor{blackGreen}{cmyk}{0.85,0,0.33,0}
\definecolor{RawSienna}{cmyk}{0,0.72,1,0.45}
\definecolor{Magenta}{cmyk}{0,1,0,0}
\definecolor{wood}{RGB}{139,115,85}
\definecolor{dorange}{RGB}{255,127,0}
\definecolor{dolive}{RGB}{85,107,47}
\definecolor{drg}{RGB}{255,165,0}
\definecolor{dgreen}{rgb}{0,.5,0}
\definecolor{dblue}{rgb}{0,0,.5}
\definecolor{dred}{rgb}{0.5,0,.5}

\newcommand{\tcr}{\textcolor{red}}
\newcommand{\tcb}{\textcolor{blue}}
 
\usepackage{graphicx}% Include figure files
\usepackage{rotating}
\usepackage{tikz}
\usepackage{gnuplottex}
\usepackage{float}

\usepackage{xspace}

\usepackage{titlesec}
\titlespacing*{\section}{0pt}{2ex}{2ex}
\titlespacing*{\subsection}{0pt}{2ex}{2ex} 
\titlespacing*{\subsubsection}{0pt}{2ex}{2ex}

\titleformat*{\section}{\large\bfseries}
\titleformat*{\subsection}{\large\bfseries}
\titleformat*{\subsubsection}{\large\bfseries}
\titleformat*{\paragraph}{\large\bfseries}
\titleformat*{\subparagraph}{\large\bfseries}

\usepackage[colorlinks,citecolor=blue,urlcolor=blue,linkcolor=blue]{hyperref}
\usepackage{cite}

\usepackage{colortbl}
\usepackage{dcolumn}% Align table columns on decimal point
\usepackage{multirow}
\usepackage{longtable}
\usepackage{arydshln}
\usepackage{array}
\usepackage{tabularx, booktabs}
\newcolumntype{Y}{>{\centering\arraybackslash}X}

\usepackage{mathtools}
\usepackage{amsmath}
\usepackage{amsfonts}
\usepackage{amssymb} 
\usepackage{amsthm}
\usepackage{mathrsfs}
\usepackage{dsfont}
\usepackage{bm}% bold math
\usepackage{bbm}
\DeclareMathAlphabet{\mathpzc}{OT1}{pzc}{m}{it}

\usepackage{stmaryrd}

\makeatletter
\newcommand\xLongLeftRightArrow[2][]{%
  \ext@arrow 0099{\LongLeftRightArrowfill@}{#1}{#2}}
\def\LongLeftRightArrowfill@{%
  \arrowfill@\Leftarrow\Relbar\Rightarrow}
\makeatother

\usepackage{simplewick}
\makeatletter
\ExplSyntaxOn
\RenewDocumentCommand{\acontraction}{ O{1ex} O{black} m m m m}{%
  \tl_if_empty:nTF{#1}{\def\wickoffset{1ex}}{\def\wickoffset{#1}}%
  {\color{#2}%
  \mathchoice
    {\acontraction@\displaystyle{#3}{#4}{#5}{#6}{\wickoffset}}%
    {\acontraction@\textstyle{#3}{#4}{#5}{#6}{\wickoffset}}%
    {\acontraction@\scriptstyle{#3}{#4}{#5}{#6}{\wickoffset}}%
    {\acontraction@\scriptscriptstyle{#3}{#4}{#5}{#6}{\wickoffset}}}}%
\RenewDocumentCommand{\bcontraction}{ O{1ex} O{black} m m m m}{%
  \tl_if_empty:nTF{#1}{\def\wickoffset{1ex}}{\def\wickoffset{#1}}%
  {\color{#2}%
  \mathchoice
    {\bcontraction@\displaystyle{#3}{#4}{#5}{#6}{\wickoffset}}%
    {\bcontraction@\textstyle{#3}{#4}{#5}{#6}{\wickoffset}}%
    {\bcontraction@\scriptstyle{#3}{#4}{#5}{#6}{\wickoffset}}%
    {\bcontraction@\scriptscriptstyle{#3}{#4}{#5}{#6}{\wickoffset}}}}%
\ExplSyntaxOff
\makeatletter

\newcommand\reallywidehat[1]{%
\savestack{\tmpbox}{\stretchto{%
  \scaleto{%
    \scalerel*[\widthof{\ensuremath{#1}}]{\kern-.6pt\bigwedge\kern-.6pt}%
    {\rule[-\textdepth/2]{1ex}{\textdepth}}%WIDTH-LIMITED BIG WEDGE
  }{\textdepth}% 
}{0.5ex}}%
\stackon[1pt]{#1}{\tmpbox}%
}

\usepackage{physics} 
\usepackage{braket}

\newtheorem{proposition}{Proposition}[section]
\newtheorem{definition}{Definition}[section]

\newtheorem{example}{Example}[section]
\newtheorem{remark}{Remark}[section]

\usepackage{algorithm}
\usepackage[noend]{algpseudocode}
\usepackage{listings}
\usepackage{diagbox}

\if False
\documentclass[multi=my,crop]{standalone}

\usepackage{mathtools}
\usepackage[utf8]{inputenc}
\usepackage[french]{babel}
\usepackage[T1]{fontenc}
\usepackage{amsmath}
\usepackage{amsfonts}
\usepackage{bm}
\usepackage{stmaryrd}
\usepackage{amssymb}
\usepackage{mathrsfs}
\usepackage{bm}
\usepackage{graphicx}
\usepackage{ragged2e}
\usepackage{fancybox}
\usepackage{bbm}

\usepackage{tikz}
\usetikzlibrary{backgrounds}
\usetikzlibrary{positioning}

\pgfdeclarelayer{bg}    % declare background layer
\pgfsetlayers{bg,main}
\usepackage{colortbl}
\usepackage{multimedia}
\usepackage{xcolor}
\usepackage{braket}

\usepackage{amsfonts}
\usepackage{ifthen}

\fi

\newcommand{\padaga}[5]{
	\def\x{#1}
	\def\e{.015}%demi de l'épaisseur du trait
	\def\r{.5}
	
	\ifthenelse{\equal{\detokenize{#2}}{\detokenize{occ}}}
		{ \draw[fill] (\x,-\e) arc (270:90:\e+\r) node[above] {\Large $\hat{a}^\dag_{#3}$}
			-- (\x,1+\e) arc (90:270:\r) 
			-- cycle; }
		{ \ifthenelse{\equal{\detokenize{#2}}{\detokenize{vir}}}	
			{ \draw[fill] (\x,\e)
				-- (\x-2*\e,\e)
				-- (\x-2*\e,\e)
				-- (\x-2*\e,-\e) arc (0:-90:2*\r) node[below] {\Large $\hat{a}^\dag_{#3}$}
				-- (\x-2*\r,-2*\r-\e) arc (-90:0:2*\r) --cycle;}
			{ \draw (\x-1,0) node[above] {\Large $\hat{a}^\dag_{#3}$};
			  \def\x{#1-1}}
		}
	
	\draw[fill] (\x,\e) 
		-- (\x+2,\e) 
		-- (\x+2,-\e) 
		-- (\x,-\e) 
		-- cycle;
	
	\ifthenelse{\equal{\detokenize{#4}}{\detokenize{occ}}}
		{ \draw[fill] (\x+2,\e) arc (-90:90:\r)  node[above] {\Large $\hat{a}_{#5}$}
			-- (\x+2,1+\e) arc (90:-90:\e+\r) 
			-- cycle; ;}
		{\ifthenelse{\equal{\detokenize{#4}}{\detokenize{vir}}}
			{\draw[fill] (\x+2,\e) 
			-- (\x+2*\e+2,\e) 
			-- (\x+2*\e+2,\e)
			-- (\x+2*\e+2,-\e) arc (180:270:2*\r) node[below] {\Large $\hat{a}_{#5}$}
			-- (\x+2*\r+2,-2*\r-\e) arc (270:180:2*\r) --cycle;}
			{\draw (\x+2,0) node[above] {\Large $\hat{a}_{#5}$};}
		}	
}

\newcommand{\adaga}[4]{%one-body operator
  \ifthenelse{\equal{\detokenize{#1}}{\detokenize{exci}}}
   { \draw[line width=1pt] (0,-1) arc(-90:0:1);
    \tikzset{shift={(1,0)}} }{}
  \ifthenelse{\equal{\detokenize{#1}}{\detokenize{deex}}}
   { \draw[line width=1pt] (0,1) arc(90:270:.5) -- (1,0);
    \tikzset{shift={(1,0)}} }{}

   \draw[line width=1pt] (0,0)node[above]{\Large$\hat{a}^\dag_{#2}$} -- (1.5,0)node[above]{\Large$\hat{a}_{#4}$};
   \tikzset{shift={(1.5,0)}}

  \ifthenelse{\equal{\detokenize{#3}}{\detokenize{deex}}}
   { \draw[line width=1pt] (0,0) arc(180:270:1);
    \tikzset{shift={(1,0)}} }{}
  \ifthenelse{\equal{\detokenize{#3}}{\detokenize{exci}}}
   { \draw[line width=1pt] (0,0) -- (1,0) arc(-90:90:.5);
    \tikzset{shift={(1,0)}} }{}

}

\if False

\usetikzlibrary{decorations.markings}
\usetikzlibrary { decorations.pathmorphing, decorations.pathreplacing, decorations.shapes, }
\usepackage{pgfplots}

\usetikzlibrary{decorations.pathmorphing}
\usetikzlibrary{decorations.markings}
\usetikzlibrary{arrows.meta,bending}

\tikzset{grid/.style={step=1cm,gray,very thin}}

\fi

\makeatletter
\newcommand\footnoteref[1]{\protected@xdef\@thefnmark{\ref{#1}}\@footnotemark}
\makeatother

\begin{document} 

\title{\normalsize{\textbf{Dyck language and fermionic second quantization}} \\ \normalsize{\textbf{I. Theory}}}

\date{} 

\maketitle

\vspace*{-2cm}

\noindent \begin{center}
\textit{by} Jérémy Morere\footnote{\label{note1}\textit{Université de Lorraine,} CNRS, LPCT, \textit{F-54000 Nancy, France}}{}$^,$\footnote{jeremy.morere@univ-lorraine.fr} \textit{and} Thibaud Etienne\footnoteref{note1}{}$^,$\footnote{thibaud.etienne@univ-lorraine.fr}\\
\end{center}

$\;$

\begin{abstract}
\noindent This paper proposes a novel framework connecting fermionic second quantization and Dyck languages.
By defining translations of creation and annihilation operators using bracket alphabets, the study establishes nullity criteria for expectation values of chains of second quantization operators. Those translations are designed to reveal sufficient conditions for the nullity of the expectation values relatively to one-determinant states or to the physical vacuum. The nullity criteria are purely syntactic, and simply reduce to the inspection of sequences of opening and closing brackets.
Moreover, numbers and transformations in Dyck languages can be imported in the context of fermionic second quantization. One of these numbers, the depth, originally absent from second quantization, can be used to introduce more nullity criteria.
\\ $\;$ \\ 
\textit{Keywords: Fermionic second quantization, Dyck language, Diagrammatic representation.}
\end{abstract}

\section{Introduction}

The formalism of second quantization \cite{dirac_quantum_1927,szabo_modern_2012} plays a central role in the theoretical description of quantum many-body systems, particularly in condensed matter physics and quantum chemistry. By expressing physical observables and quantum states in terms of creation and annihilation operators living in the Fock space, second quantization provides a compact and consistent framework to treat bosonic or fermionic systems.
This formalism is present in a wide range of electronic structure methods, including Configuration Interaction (CI) \cite{sherrillt_configuration_nodate,helgaker_molecular_2014, szabo_modern_2012}, Coupled Cluster (CC) theory \cite{helgaker_molecular_2014, cizek_correlation_1966, shavitt_many-body_2009}, and many-body perturbation theory such as Møller--Plesset\cite{helgaker_molecular_2014,cremer_mollerplesset_2011,moller_note_1934}.

The expectation value of a chain of fermionic creation and annihilation operators relatively to the physical vacuum or relatively to a one-determinant state is either zero, one, or minus one. This is true regardless of the sequence of second quantization operators of the chain.
 
In most of the practical post-Hartree Fock methods, a truncation of the method's central expansion(s) is introduced based solely on the excitation rank. In that context, all the chains of second quantization operators involved are known \textit{a priori} and can be pre-computed in order to save computation time.
Another approach is to introduce an adaptive truncation, like in selected-CC \cite{evangelista_adaptive_2014, lyakh_adaptive_2010}.
The set of determinant states used in these approaches is updated at each iteration. A recurring difficulty in these approaches is the impossibility to predict which expectation value will have to be computed on-the-fly.
Therefore, a lot of computation time is actually spent on the evaluation of expectation values that are equal to zero.

Symbolically determining the expectation value of chains of fermionic second quantization operators becomes increasingly challenging as the length of the operator chain grows, particularly in methods involving nested commutators such as CC theory. The number of terms generated during such evaluations grows combinatorially, leading to substantial computational overhead.
One traditional solution to this problem relies on Wick’s theorem \cite{wick_evaluation_1950,lindgren_atomic_1986}, which rewrites operator chains in terms of normal-ordered chains and contractions. While Wick’s theorem transforms the problem into a combinatorial enumeration of pairings, it still requires identifying all valid contractions, a task whose complexity grows rapidly with the chain size.

In parallel, combinatorics have provided valuable insights into the structure of normal-ordering. Notably, the work of Blasiak and Flajolet \cite{blasiak_combinatorial_2011} introduced a combinatorial interpretation of creation and annihilation operators through graph-based representations, highlighting the deep connections between operator ordering and combinatorial structures. 
These observations motivate the search for purely syntactic and combinatorial sufficient conditions for the nullity of the expectation value of a chain of fermionic second quantization operators without explicitly applying the operators to a state, or enumerating contractions. Such a criterion would significantly reduce computational complexity.

Using a diagrammatic representation in order to make the manipulation of second quantization easier is not new. Since the introduction of the Feynam diagram \cite{feynman_qed_2006,mattuck_guide_1992}, diverse diagrammatic representations of problems involving second quantization operators are used. However, each diagram is relevant for a specific method such as Goldstone\cite{goldstone_derivation_1957,mattuck_guide_1992,szabo_modern_2012} and Hugenholtz\cite{hugenholtz_perturbation_1957,mattuck_guide_1992,szabo_modern_2012} diagrams for general perturbation theory, or coupled-cluster diagrams \cite{shavitt_many-body_2009} for CC.
Despite their effectiveness, these diagrammatic approaches are typically tailored to specific methods and do not provide a universal, method-independent criterion for determining whether a given operator chain yields a non-zero expectation value.
Moreover, the existence of sufficient conditions for the nullity of the expectation value of a chain of fermionic second quantization operators is not new either. The Slater--Condon rules \cite{slater_theory_1929,condon_theory_1930} are examples of such criteria.

In this work, we first intend to build a one-dimensional representation of any chain of fermionic second quantization operators which can reveal sufficient conditions for the nullity of the expectation value of the chain relatively to some reference states.
We address this question by establishing a connection between fermionic operator chains and Dyck language\cite{stanley_catalan_2015,roman_introduction_2015}, a well-known object in combinatorics. 
Dyck words are balanced sequences of brackets that are counted by Catalan numbers \cite{stanley_catalan_2015,roman_introduction_2015} and appear in a wide range of mathematical contexts. In particular, Catalan numbers can be found in diverse papers related to second quantization \cite{accardi_non--commutative_1994,anshelevich_partition-dependent_2001,arienzo_bosonic_2025,dumitrescu_invitation_2017,effros_feynman_2003,ginot_large_2022,jackson_robust_2017}

The present article comes together with some work reported in a companion paper (hereafter denoted by ``Part II'') which shows how our new framework can be extended and used for an efficient application of a corollary to the static-fermionic Wick theorem, or for dealing with (nested) commutators, together with some numbering and some algorithmic considerations.

This first article is organized as follows: After a brief introduction to the fermionic second quantization formalism in Section II, we present nullity criteria for the expectation value of chains of fermionic second quantization operators based on the sequence of creation and annihilation operators. 

In Section III, we introduce what is a Dyck word in a Dyck language, how it can be described and manipulated. We will highlight a specific transformation called \textit{expulsion}, this transformation will be studied in detail in Part II. 

Section IV connects fermionic second quantization operators and Dyck words, by defining different ``translations'' of fermionic second quantization operators into different bracket-based alphabets. Depending on the context, we explain which translation is relevant for revealing sufficient conditions for the nullity of the expectation value of the operator chain from the sequence of opening and closing brackets in its translation. In this first part, there is a loss of information between the chain of fermionic second quantization and its corresponding translation. In Part II we recover the missing information.

The present paper closes with some perspectives about bosonic second quantization, and about the representation of transformations which do not necessarily conserve the number of particles, using an alternative diagrammatic representation based on so-called Dyck paths. This last perspective could be of interest in the framework of charged excitation calculations.

\section{Basics of fermionic second quantization}
\noindent We provide here a brief introduction to fermionic second quantization in orthonormal basis. For more details we strongly advise the reader to consult References \cite{szabo_modern_2012, surjan_second_1989, lindgren_atomic_1986}.

Since in Part I and in Part II we are concerned with \textit{fermionic} second quantization, when there is no explicit specification it is assumed that ``second quantization operator'' implicitely points at \textit{fermionic} second quantization operators.

\subsection{Definitions, notations and labelling}

We consider an $L$-dimensional orthonormal ordered basis $\mathcal{B} \coloneqq(\varphi_r)_{ 1\leq r \leq L }$ of one-particle wave functions, called spin-orbitals.

Let $N$ be a natural integer lower than $L$. Let $(\varphi_{i_j})_{1\leq j\leq N}$ be a sequence of $N$ elements of $\mathcal{B}$. The one-determinant state made of these elements in that order is defined as
$$\ket{\varphi_{i_1}\cdots \varphi_{i_N}} \coloneqq \bigwedge_{j=1}^N \varphi_{i_j},$$
i.e., the totally antisymmetrized tensor product of elements of the chosen sequence. Notice that $\ket{\varphi_{i_1}\cdots \varphi_{i_N}}$ is readily normalized to unity. The physical vacuum state, denoted $\ket{\;}$, describes no electron but is normalized to unity. The zero-electron Fock space, $\mathcal{F}_0$, is defined as
\begin{equation*}
\mathcal{F}_0 \coloneqq \left\lbrace \lambda\ket{\;}\,:\, \lambda\in\mathbb{C}\right\rbrace.
\end{equation*}
The zero of $\mathcal{F}_0$, denoted by $\ket{0_0}$, is simply $0\ket{\;}$. We also define $\mathcal{F}_1$ as
\begin{equation*}
\mathcal{F}_1 \coloneqq \mathrm{span}_{\mathbb{C}}(\mathcal{B}) = \left\lbrace\sum_{i=1}^L\lambda_i\varphi_i \, : \, \forall i \in \llbracket 1, L \rrbracket, \, \lambda_i\in\mathbb{C}\right\rbrace
\end{equation*}
where $\llbracket 1, L\rrbracket$ is the natural integer interval starting with 1 and ending with $L$, including both 1 and $L$. The zero of $\mathcal{F}_1$ is denoted $\ket{0_1}$. For every $k$ in $\llbracket 2, L\rrbracket$, the $k$--electron Fock space, $\mathcal{F}_k$, is defined as
\begin{equation*}
\mathcal{F}_k\coloneqq \mathrm{span}_{\mathbb{C}}\left\lbrace \bigwedge_{j=1}^k \varphi_{i_j} \, : \, 1 \leq i_1 \leq \cdots \leq i_k\leq L\right\rbrace
\end{equation*}
with its zero being denoted below by $\ket{0_k}$.

To each spin-orbital $\varphi_r$ we associate a creation operator $\hat{a}^\dag_r$, which creates an electron described by $\varphi_r$ when acting on any determinant state:
\begin{equation*}
\hat{a}^\dag_r \ket{\varphi_{i_1} \cdots  \varphi_{i_N}} = \ket{\varphi_r \varphi_{i_1} \cdots \varphi_{i_N}}.
\end{equation*} 
If $\varphi_r$ is already present in the determinant state, the resulting state is the zero of the $(N+1)$--electron Fock space, i.e., $\ket{0_{N+1}}$
\begin{equation} \label{eq:zero_creation}
r\in(i_1,\ldots,i_N)\,\Longrightarrow \,\hat{a}_r^\dag \ket{\varphi_{i_1} \cdots \varphi_{i_N}} = \ket{0_{N+1}}.
\end{equation}
Similarly, the annihilation operator $\hat{a}_r$ removes an electron from $\varphi_r$ if that spin-orbital is present in the determinant state:
\begin{align}
\hat{a}_r \ket{\varphi_r \varphi_{i_1} \cdots \varphi_{i_N}} &= \ket{\varphi_{i_1} \cdots \varphi_{i_N}}.
\end{align}
If $\varphi_r$ is absent, the result is the zero of the $(N-1)$--electron Fock space, i.e., $\ket{0_{N-1}}$
\begin{equation} \label{eq:zero_annihilation}
r\notin(i_1,\ldots,i_N)\,\Longrightarrow \,\hat{a}_r \ket{\varphi_{i_1} \cdots \varphi_{i_N}} = \ket{0_{N-1}}.
\end{equation}
Applying any annihilation operator to the physical vacuum returns the null state:
\begin{equation}\label{eq:zero_annihilationPhysical}
\forall r \in \llbracket 1, L \rrbracket,\, \hat{a}_r\ket{\;} = \ket{0_0}.
\end{equation}
Applying any creation operator to an element of $\mathcal{F}_L$ returns the zero of $\mathcal{F}_L$
\begin{equation}\label{eq:zero_creation_FL}
\forall \ket{\Psi} \in \mathcal{F}_L,\,\forall r \in \llbracket 1, L \rrbracket,\, \hat{a}_r^\dag\ket{\Psi} = \ket{0_L}.
\end{equation}
Let $r$ and $s$ be two integers between 1 and $L$. The creation and annihilation operators satisfy some anticommutation relations: two relatively to the anticommutation of two creation operators, namely
\begin{equation*}
\forall i\in\llbracket 0 , (L-2)\rrbracket, \forall \ket{\Psi} \in \mathcal{F}_i,\left[\hat{a}^\dag_r,\hat{a}^\dag_s\right]_+\ket{\Psi} = \ket{0_{i+2}},
\end{equation*}
and
\begin{equation*}
\forall i\in\llbracket (L-1) , L \rrbracket, \forall \ket{\Psi} \in \mathcal{F}_i,\left[\hat{a}^\dag_r,\hat{a}^\dag_s\right]_+\ket{\Psi} = \ket{0_{L}},
\end{equation*}
two relatively to the anticommutation of two annihilation operators, namely
\begin{equation*}
\forall i\in\llbracket 0 , 1 \rrbracket, \forall \ket{\Psi} \in \mathcal{F}_i,\left[\hat{a} _r,\hat{a} _s\right]_+\ket{\Psi} = \ket{0_{0}},
\end{equation*}
and
\begin{equation*} 
\forall i\in\llbracket 2 , L\rrbracket, \forall \ket{\Psi} \in \mathcal{F}_i,\left[\hat{a}_r,\hat{a}_s\right]_+\ket{\Psi} = \ket{0_{i-2}},
\end{equation*} 
and finally the anticommutation relation implying one creation and one annihilation operator:
\begin{equation}
\left[\hat{a}^\dag_r,\hat{a}_s\right]_+ = \delta_{r, s}\hat{\mathds{I}}.\label{eq:anticommutation_with_identity}
\end{equation}
In \eqref{eq:anticommutation_with_identity}, ``$\hat{\mathds{I}}$'' is the identity operator.

\begin{remark}
From now on, we will only consider second quantization operators indexed relatively to elements of $\mathcal{B}$, and spin-orbital functions that are elements of $\mathcal{B}$.
\end{remark}
\noindent Any determinant state can be generated from the physical vacuum by successive application of creation operators,
\begin{equation*}
\ket{\varphi_{i_1} \cdots \varphi_{i_N}} = \hat{a}^\dag_{i_1} \cdots \hat{a}^\dag_{i_N}\ket{\;}.
\end{equation*}
\begin{definition}[Occupation number relatively to the physical vacuum]
The occupation number of any spin-orbital $\varphi_r$ relatively to the physical vacuum $\ket{\;}$ is defined as $\braket{\;|\hat{a}^\dag_r\hat{a}_r|\;}$, which is always equal to zero.
\end{definition}
\begin{definition}[Occupation number relatively to a one-determinant state]
The occupation number of a spin-orbital $\varphi_r$ relatively to a one-determinant state $\Psi$ is a number, written $n_r^\Psi$, that is equal to one if $\varphi_r$ is present in the $\Psi$ determinant state and zero otherwise. 
\end{definition}
\noindent Instead of the physical vacuum, one may adopt as a reference a one-determinant state, known as the Fermi vacuum $\Psi_0$. Usually, $\mathcal{B}$ is ordered such that its $N$ first elements are the only elements of the Fermi vacuum:
%In this work, the Fermi vacuum $\ket{\Psi_0}$ contains $N$ electrons, with $N \leq L$.
\begin{equation*}
\ket{\Psi_0} = \ket{{\varphi_1 \cdots \varphi_N}}.
\end{equation*}
Three cases will be encountered in this contribution: (i) the second quantization operator of interest points to a spin-orbital that is present in the Fermi vacuum (i.e., the corresponding Fermi-vacuum occupation number is one), (ii) the operator points to a spin-orbital that is absent in the Fermi vacuum (i.e., the corresponding occupation number is zero), or (iii) the occupation number of the spin-orbital to which the operator points is not specified. For each case, a specific notation is introduced: ``$o$'' is used for pointing at \textit{occupied} spinorbitals, i.e. in case (i); ``$v$'' is used for pointing at \textit{virtual} spinorbitals, i.e., in case (ii); and ``$a$'' is used for \textit{arbitrarily} pointing at any element of $\mathcal{B}$, i.e., in case (iii).  These three cases are summarized in Table \ref{tab:occ-virt-arb}.
\begin{table}[h!]
\begin{center}
\begin{tabular}{cccc}
\hline 
 & Operators & Indexes & Occupation number \\ 
\hline 
Occupied & $\hat{o}^\dag,\hat{o}$ & \textit{i,j,k},$\ldots$ & 1 \\ 
Virtual & $\hat{v}^\dag,\hat{v}$ & \textit{a,b,c},$\ldots$ & 0 \\ 
Arbitrary & $\hat{a}^\dag,\hat{a}$ & \textit{r,s,t},$\ldots$ & 1 or 0 \\ 
\hline 
\end{tabular} 
\caption{Summary of the cases encountered in this paper for labelling the second quantization operators relatively to the Fermi vacuum.}
\label{tab:occ-virt-arb}
\end{center}
\end{table}
For the sake of brevity, in the following we name \textit{v-}creation and \textit{v-}annihilation operators the second quantization operators corresponding to virtual spin-orbitals. Similar convention is used for introducing \textit{o-}creation and \textit{o-}annihilation operators for operators corresponding to occupied spin-orbitals. We also use ``\textit{o}-operator'' and ``\textit{v-}operator'' expressions when the status (creation/annihilation) is not specified.

Relative to the Fermi vacuum, we define excitation operators which combine a creation operator on a virtual orbital with an annihilation operator on an occupied orbital: Let $i$ be any integer between 1 and $N$, and let $a$ be any integer between $(N+1)$ and $L$. The excitation operator corresponding to the removal of $\varphi_i$ and its replacement by $\varphi_a$ reads 
\begin{equation}\label{eq:ExcOp}
\hat{E}_i^a\coloneqq\hat{v}_a^\dag\hat{o}_i.
\end{equation}
For the sake of clarity in further notations, we also introduce the corresponding deexcitation operator, 
\begin{equation}\label{eq:DeexcOp}
\hat{D}_a^i\coloneqq\hat{o}_i^\dag\hat{v}_a,
\end{equation}
which combines a creation operator on an occupied orbital with an annihilation operator on a virtual orbital.

\subsection{Nullity criteria for expectation values}

\noindent Let $\hat{C}$ be a chain of second quantization operators. Let $\Psi$ be a quantum state. In the following, we name \textit{expectation value of} $\hat{C}$ \textit{relatively to} $\Psi$ the $\braket{\Psi,\hat{C}\Psi}$ quantity, that we write in the physicists bra-ket notation $\braket{\Psi|\hat{C}|\Psi}$ or, in short, $\braket{\hat{C}}_{\Psi}$. This extends to the physical vacuum, in which case the short notation will simply be $\braket{\hat{C}}$.

The scalar product of two $N$-function determinant states $\Psi_1$ and $\Psi_2$ both written solely with elements of $\mathcal{B}$ and differing by at least one spin-orbital is zero:
\begin{equation}\label{eq:orthogonalite_SD}
(\exists r \in \llbracket 1, L\rrbracket,\, n_r^{\Psi_1} \neq n_r^{\Psi_2}) \, \Longleftrightarrow \, \braket{\Psi_1|\Psi_2} = 0.
\end{equation}

\subsubsection{Expectation values relatively to the physical vacuum}

\noindent If the $\hat{C}$ chain is a single second quantization operator, its expectation value relatively to the physical vacuum is zero. Indeed, it either annihilates into the bra or ket physical vacuum. A generalization of this simple observation can be done for any $\hat{C}$ chain. 

\begin{proposition}\label{prop:criteria_nullity_EVC_2nd}
Let $\hat{C}$ be a fermionic second quantization operators chain. The expectation value of $\hat{C}$ relatively to the physical vacuum is zero if at least one of the two following criteria is not met:
\begin{itemize}
\item[P1] Reading right to left, there {is} no more annihilation than creation operators,
\item[P2] The number of creation and annihilation operators are equal.
\end{itemize}
\end{proposition}

\begin{proof}[Proof for P1]
Let us consider the $M$ first operators of $\hat{C}$ read from the right to the left. Amongst these $M$ operators, there are $m_c$ and $m_a$ creation and annihilation operators, respectively. We need to study the case $m_c < m_a$. Applying the $M$ operators to $\ket{\;}$, if situation described in \eqref{eq:zero_creation} or \eqref{eq:zero_annihilation} is met the $\braket{\hat{C}}$ expectation value is zero. Otherwise, we expect to find $(m_c - m_a)$ spin-orbitals involved in the determinant state, which is strictly less than zero. This implies that at a certain point, an annihilation operator has been applied to the physical vacuum, leading to the null state according to \eqref{eq:zero_annihilationPhysical}, and to an expectation value of zero.
\end{proof}

\begin{proof}[Proof for P2]
Let $M_c$ and $M_a$ be the number of creation and annihilation operators in $\hat{C}$, respectively. If $M_c > M_a$, unless we are in the case described in \eqref{eq:zero_creation} or \eqref{eq:zero_annihilation} — in which case $\braket{\hat{C}}$ is readily zero —, there exists at least one $r$ value in $\llbracket 1,L\rrbracket$ such that the occupation number corresponding to $\varphi _r$ in $\hat{C}\ket{\;}$ is different from zero. Accordingly, and due to \eqref{eq:orthogonalite_SD}, we find $\braket{\hat{C}}=0$. If $M_c < M_a$, we know from the above proof that $\hat{C}\ket{\;} = \ket{0_0}$.
\end{proof}
\noindent {\normalfont{{Proposition \ref{prop:criteria_nullity_EVC_2nd}}}} is pure-syntactic: The determination of whether a given chain satisfies sufficient conditions for its expectation value relatively to the physical vacuum to be zero reduces to analysing the structure of the chain. We would like to highlight that this proposition should be understood as a \textit{negative} criterion of nullity of $\braket{\hat{C}}$, which is decided if at least one of the two criteria is \textit{not} met. We emphasize this point with the two following examples, which explore the case of chains of operators that satisfy only one criterion among the two defined in {\normalfont{{Proposition \ref{prop:criteria_nullity_EVC_2nd}}}} but not the other.

\begin{example}
Let $(q,r,s,t)$ be a 4-tuple of integers, all belonging to $\llbracket 1,L\rrbracket$. Then,
\begin{equation*}
\bra{\;} \hat{a}_q^\dag\hat{a}_r\hat{a}_s\hat{a}^\dag_t\ket{\;} = \bra{\;} \hat{a}_q^\dag\hat{a}_r\hat{a}_s\ket{\varphi_t} = \delta_{s, t}\bra{\;} \hat{a}_q^\dag\hat{a}_r\ket{\;} = \delta_{s,t}\left(0\bra{\;} \hat{a}_q^\dag\ket{\;}\right) = 0.
\end{equation*}
Criterion 2 in {\normalfont{{Proposition \ref{prop:criteria_nullity_EVC_2nd}}}} is satisfied, but not criterion 1. Indeed, consider $M=3$ in proof of {\normalfont{{Proposition \ref{prop:criteria_nullity_EVC_2nd}}}}: two annihilation operators and one creation operator are met. Hence the expectation value can be deduced to be zero without having to apply the operators one after the other.
\end{example}

\begin{example}
Let $(q,r,s)$ be a 3-tuple of integers, all belonging to $\llbracket 1,L\rrbracket$. Then,
\begin{equation*}
\bra{\;}\hat{a}_q^\dag\hat{a}_r\hat{a}^\dag_s\ket{\;} = \delta_{r, s}\braket{\;|\varphi_q} = 0.
\end{equation*}
In that case, we see that criterion 1 in {\normalfont{{Proposition \ref{prop:criteria_nullity_EVC_2nd}}}} is satisfied, but criterion 2 is not. Again, the expectation value can be deduced to be zero without having to apply the operators one after the other.
\end{example}

\begin{remark}
It is important to notice that if one criterion is satisfied, it says nothing about whether the $\braket{\hat{C}}$ expectation value is zero or not: not meeting at least one of the two criteria is sufficient but not necessary to ensure a zero expectation value. This is illustrated in  {\normalfont{Example \ref{ex:PhysVac3}}}.
\end{remark}

\begin{example} \label{ex:PhysVac3}
Let $p$ and $q$ be two different integers, both belonging to $\llbracket 1,L\rrbracket$. Then,
 $$\braket{\;|\hat{a}_p\hat{a}_q^\dag|\;}=0.$$
 The structure of the chain satisfies both criteria in {\normalfont{{Proposition \ref{prop:criteria_nullity_EVC_2nd}}}} but the expectation value is zero due to the fact that $p$ is different from $q$.
\end{example} 
  
\subsubsection{Expectation values relatively to the Fermi vacuum}

\noindent If, rather than the physical vacuum, one is interested in evaluating expectation values relatively to the Fermi vacuum, the \textit{o}- and \textit{v-}operators must be considered separately. Indeed, while \textit{o-}creation operators applied to the Fermi vacuum results in $\ket{0_{N+1}}$ --- see \eqref{eq:zero_creation} ---, we see that acting with a \textit{v-}creation operator on the Fermi vacuum results in adding the target spin-orbital to the determinant state. On the other hand, applying an \textit{o-}annihilation operator to the Fermi vacuum will remove the target spin-orbital from the determinant state while applying a \textit{v-}annihilation operator to the Fermi vacuum will lead to $\ket{0_{N-1}}$ --- see \eqref{eq:zero_annihilation}.

Proposition \ref{prop:criteria_nullity_EVC_fermi_2nd} is the transposition of Proposition \ref{prop:criteria_nullity_EVC_2nd} to expectation values relatively to the Fermi vacuum.

\begin{proposition}\label{prop:criteria_nullity_EVC_fermi_2nd}
Let $\hat{C}$ be a fermionic second quantization operators chain. The expectation value of $\hat{C}$ relatively to the Fermi vacuum is zero if at least one of the four following criteria is not met:
\begin{itemize}
\item[P1] Reading right to left, there is no more \textit{v-}annihilation than \textit{v-}creation operators.
\item[P2] The number of \textit{v-}creation and \textit{v-}annihilation operators are equal.
\item[P3] Reading right to left, there is no more \textit{o-}creation than \textit{o-}annihilation operators.
\item[P4] The number of \textit{o-}creation and \textit{o-}annihilation operators are equal.
\end{itemize}
\end{proposition}

\begin{proof}[Proof for P1 and P2]
The operators in $\hat{C}$ can be reordered in two ways in order to provide two new chains, namely $\hat{C}_o$ and $\hat{C}_v$. The former is structured as follows: reading it from the right to the left, all the \textit{o}-operators are first met in the same order as in $\hat{C}$; all the \textit{v-}operators are then met, in the same order as in $\hat{C}$. The latter, i.e., $\hat{C}_v$, is constructed in a similar fashion: all the \textit{v}-operators are first met in the same order as in $\hat{C}$, before the \textit{o}-operators are met, in the same order as in $\hat{C}$. The only effect of such a reorganization of the operators can be a change in the sign of the expectation value. Therefore, if the expectation value of $\hat{C}_o$ (respectively, $\hat{C}_v$) relatively to the Fermi vacuum is equal to zero, so is $\braket{\hat{C}}_{\Psi_0}$. An illustration of these rearrangements is given in Example \ref{ex:fermi_vac}. 

All \textit{v-}operators behave the same way relatively to the Fermi and the physical vacuum. If \textit{v-}operators do not respect at least one of the two criteria introduced in {\normalfont{Proposition \ref{prop:criteria_nullity_EVC_2nd}}} the expectation value of $\hat{C}_v$ is equal to zero, so is $\braket{\hat{C}}_{\Psi_0}$.
\end{proof}

\begin{proof}[Proof for P3]
Let us consider the $M$ first \textit{o-}operators of $\hat{C}_o$ read from the right to the left. Amongst these $M$ operators, there are $m_c$ and $m_a$ \textit{o-}creation and \textit{o-}annihilation operators, respectively. We need to study the case $m_a < m_c$. Applying the $M$ \textit{o-}operators to $\ket{\Psi_0}$, if situation described in \eqref{eq:zero_creation} or \eqref{eq:zero_annihilation} is met the $\braket{\hat{C}_o}_{\Psi_0}$ expectation value is zero. Otherwise, we expect to find $(N + m_c - m_a)$ occupied spin-orbitals involved in the determinant state, which is strictly more than $N$. This implies that at a certain point, an \textit{o-}creation operator has been applied to the Fermi vacuum, resulting in the zero of the $(N+1)$--electron Fock space because all occupied spin-orbitals are already in the Fermi vacuum, and to a value of zero for $\braket{\hat{C}_o}_{\Psi_0}$, in which case $\braket{\hat{C}}_{\Psi_0}$ is also equal to zero.
\end{proof}

\begin{proof}[Proof for P4]
Let $M_c$ and $M_a$ be the number of \textit{o-}creation and \textit{o-}annihilation operators in $\hat{C}_o$, respectively. If $M_a > M_c$, unless we are in the case described in \eqref{eq:zero_creation} or \eqref{eq:zero_annihilation} — in which case $\braket{\hat{C}_o}_{\Psi_0}$ is readily zero —, there exists at least one $r$ value in $\llbracket 1,N\rrbracket$ such that the occupation number corresponding to $\varphi _r$ in $\hat{C}_o\ket{\Psi_0}$ is different from one. Accordingly, and due to \eqref{eq:orthogonalite_SD}, we find $\braket{\hat{C}_o}_{\Psi_0}=0$. If $M_a < M_c$, we know from the above proof that $\braket{\hat{C}_o}_{\Psi_0}=0$. Both cases ($M_a < M_c$ and $M_a > M_c$) lead to $\braket{\hat{C}}_{\Psi_0} = 0$.
\end{proof}
\noindent Compared to the physical vacuum case, the number of criteria doubled. The role of creation and annihilation operator are flipped between \textit{o-}operator and \textit{v-}operator.
Any \textit{o-}operators anticommutes with any \textit{v-}operator. Thanks to this, the arrangement of \textit{o-}operator relatively to \textit{v-}operator does not matter, solely the structure of \textit{o-}operators taken alone and  \textit{v-}operators taken alone import. The following example illustrates this point.
\begin{example}\label{ex:fermi_vac}
Let $(i,j)$ be a couple of integers, both belonging to $\llbracket 1,N\rrbracket$. Let $(a,b)$ be a couple of integers, both belonging to $\llbracket N+1,L\rrbracket$. We build a $\hat{C}$ chain: $\hat{C}$ = $\hat{o}_i\hat{v}_a\hat{o}_j^\dag\hat{v}^\dag_b$. The $\hat{C}_v$ chain corresponding to $\hat{C}$ reads $\hat{o}_i\hat{o}_j^\dag\hat{v}_a\hat{v}^\dag_b$. We then see that
\begin{equation*}
 \braket{\hat{o}_i\hat{v}_a\hat{o}_j^\dag\hat{v}^\dag_b}_{\Psi_0} = - \braket{\hat{o}_i\hat{o}_j^\dag\hat{v}_a\hat{v}^\dag_b}_{\Psi_0} = 0.
\end{equation*}
The $\hat{C}_o$ chain corresponding to $\hat{C}$ reads $\hat{v}_a\hat{v}^\dag_b\hat{o}_i\hat{o}_j^\dag$. Then,
\begin{equation*}
 \braket{\hat{o}_i\hat{v}_a\hat{o}_j^\dag\hat{v}^\dag_b}_{\Psi_0} = - \braket{\hat{v}_a\hat{v}^\dag_b\hat{o}_i\hat{o}_j^\dag}_{\Psi_0} = 0.
\end{equation*}
The three chains ($\hat{C}$, $\hat{C}_v$, and $\hat{C}_o$) are equal up to sign. All of them satisfy criteria 1, 2, and 4 in {\normalfont{Proposition \ref{prop:criteria_nullity_EVC_fermi_2nd}}}, but not criterion 3. 
\end{example}
\noindent {\normalfont{Proposition \ref{prop:criteria_nullity_EVC_fermi_2nd}}} can be used even if the $\hat{C}$ chain contains only one type (\textit{v-} or \textit{o-}) of operator. This can be illustrated in {\normalfont{Example \ref{ex:Fermi_vac_no_vir}}} for a chain solely composed with \textit{o-}operators.

\begin{example}\label{ex:Fermi_vac_no_vir}
Let $(i,j,k)$ be a 3-tuple of integers, all belonging to $\llbracket 1,N\rrbracket$. Then,
\begin{equation*}
\braket{\hat{o}_i\hat{o}_j^\dag\hat{o}_k}_{\Psi_0} = 0.
\end{equation*}
Criteria 1 and 2 in {\normalfont{Proposition \ref{prop:criteria_nullity_EVC_fermi_2nd}}} are satisfied; criterion 3 is met, but criterion 4 is not.
\end{example}

\section{Dyck language}

\noindent In this section we introduce the two-symbol Dyck language, i.e., the dedicated alphabet (denoted $\mathcal{L}_1$ below) together with the corresponding syntax — i.e., the construction rules for a well-formed formula — materialized in Definition \ref{def:dyck_word} below. Afterwards, we show how one can manipulate strings written in this language in order to generate other chains. These manipulations will reveal to be useful in the next section when we will provide tools to connect Dyck language and second quantization.

\subsection{Definition and characterization of Dyck words}

Let $\mathcal{L}_1 \coloneqq \{ \bm{)}, \bm{(}\}$ be the alphabet consisting in the two following symbols: the opening round bracket ``$\bm{(}$'' and the closing round bracket ``$\bm{)}$''. Some strings written using elements of $\mathcal{L}_1$ can be termed ``Dyck words'', some cannot. Definition \ref{def:dyck_word} provides two construction rules allowing one to characterize such a string as a Dyck word.

\begin{definition}[Dyck word]\label{def:dyck_word}
Let $S$ be a string written using elements $\mathcal{L}_1$. The $S$ string is a {\normalfont{Dyck word}} if and only if both the following criteria are met:
\begin{itemize}
\item[1] Reading right to left, there is no more {\normalfont{``$\bm{(}$''}} than {\normalfont{``$\bm{)}$''}}.
\item[2] The number of {\normalfont{``$\bm{(}$''}} and {\normalfont{``$\bm{)}$''}} are equal.
\end{itemize}
\end{definition}
\noindent In a more intuitive way, a Dyck word is a well-nested string of brackets, with each opening bracket having the corresponding closing one.
\begin{remark}\label{rem:round-square-dyck}
The choice of using \textit{round brackets} is arbitrary; any kind (square, curly,...) of brackets could be chosen. In fact, any set of two symbols could be used rather than $\mathcal{L}_1$. For instance, the ``$A$'' and ``$B$'' symbols are often used in the literature.
\end{remark}
\begin{remark}\label{rem:emptyword}
The empty word is a Dyck word.
\end{remark}
\noindent Example \ref{ex:dyck_word} exhibits seven strings of same length — i.e., with the same number of symbols — constructed from $\mathcal{L}_1$. Five are Dyck words, two are not.

\begin{example}\label{ex:dyck_word}
Using the $\mathcal{L}_1$ alphabet, we can build all Dyck words of length 6, \textit{i.e.}, using 3 opening brackets and 3 closing brackets:
\begin{equation*}
\bm{(\,)\,(\,)\,(\,)}, \quad \bm{(\,(\,)\,)\,(\,)}, \quad \bm{(\,)\,(\,(\,)\,)}, \quad \bm{(\,(\,)\,(\,)\,)}, \quad \bm{(\,(\,(\,)\,)\,)}.
\end{equation*}
The following string meets criterion 1 in {\normalfont{Definition \ref{def:dyck_word}}}, but not criterion 2:
\begin{equation*}
\bm{(\,(\,)\,)\,)\,)}.
\end{equation*}
Hence, it is not a Dyck word. On the other hand, the following string meets criterion 2 in {\normalfont{Definition \ref{def:dyck_word}}} but not criterion 1 and is not a Dyck word either:
\begin{equation*}
\bm{)\,)\,(\,)\,(\,(}.
\end{equation*}
\end{example}

\noindent Now that Dyck words are well defined, we can introduce a quantitative tool which characterizes them. This tool is a descriptor called \textit{depth}.

\begin{definition}[Depth]\label{def:depth}
Let $S$ be Dyck word. Reading $S$ from the left to the right, the {\normalfont{depth}} $d_k$ of the $k^{\text{th}}$ position reached after reading the $k^{\text{th}}$ symbol is defined as the difference between the number of opening and closing brackets before reaching that position.
By convention, the depth reached before the string, i.e. $d_0$, is zero. 
\end{definition}
\noindent There is a one-to-one correspondence between a Dyck word and its list of depths. As illustrated in Example \ref{ex:DyckDepth}, the depth can be directly read from the Dyck word.
\begin{example}\label{ex:DyckDepth}
Consider the $\bm{(\,(\,)\,(\,)\,)\,(\,)}$ Dyck word. One can read it from the left to the right and annotate the depths after each symbol:
\begin{equation*}
{\textit{0}}\bm{\;(}{\it{\;1}}\bm{\;(}{\it{\;2}}\bm{\;)}{\it\;1}\bm{\;(}{\it\;2}\bm{\;)}{\it\;1}\bm{\;)}{\it\;0}\bm{\;(}{\it\;1}\bm{\;)}\it{\;0}.
\end{equation*}
Depth increases by one unit after each opening bracket and decreases by one unit after each closing bracket. The list of depths of the Dyck word $\bm{(\,(\,)\,(\,)\,)\,(\,)}$ is therefore $(0,1,2,1,2,1,0,1,0)$.
\end{example}
\begin{proposition}\label{prop:depthdycknonnegative}
The depth sequence of a Dyck word solely contains non-negative integers.
\end{proposition}
\begin{proof}(\textit{Reductio ad absurdum})\textbf{.} Let $S$ be a Dyck word. Consider the following possibility: at a given position $k$ in the word, the depth $d_k$ is strictly negative. Consequently, reading the word from the left to the right until the $k^\text{th}$ position, there is more closing than opening brackets. However, $S$ must contain the same number of opening and closing brackets due to criterion 2 in Definition \ref{def:dyck_word}. Therefore, there should be more opening than closing brackets on the right of position $k$, which would violate criterion 1 in Definition \ref{def:dyck_word}.
\end{proof}
\begin{definition}[Opening depth]
Let $S$ be a non-empty Dyck word. The list of {\normalfont{opening depths}} of $S$ is obtained by restricting its depth sequence to positions reached immediately after an opening bracket. The opening depth reached after reading the $k^{\text{th}}$ opening bracket is denoted by $d_k^\uparrow$.
\end{definition}

\begin{example}\label{ex:DyckOpeningDepth}
Let $S$ be the Dyck word from {\normalfont{Example \ref{ex:DyckDepth}}}. Its list of opening depths is $(1,2,2,1)$.
\end{example}

\noindent Dyck words can be categorized into two families, which are termed and defined in Definition \ref{def:strongweakdyckwords}.

\begin{definition}[Strong/weak dyck word]\label{def:strongweakdyckwords}
If a Dyck word never reaches a zero depth apart for the first and last depths, it is called a {\normalfont{strong Dyck word}}. Otherwise it is called a {\normalfont{{weak Dyck word}}}.
\end{definition}

\noindent Amongst the five Dyck words in Example \ref{ex:dyck_word}, the three first are weak Dyck words, while the two last are strong Dyck words.

\begin{proposition}\label{prop:weakdyck}
Any weak Dyck word can be decomposed into a concatenation of two or more strong Dyck words. 
\end{proposition}
\begin{proof}
This directly follows from the definition of weak Dyck words: splitting a weak Dyck word at positions in which the depth reaches zero results into sequences of strong Dyck words.
\end{proof}

\begin{definition}[Subword]
Each strong Dyck word obtained when decomposing a weak Dyck word as in {\normalfont Proposition \ref{prop:weakdyck}} is called a {\normalfont{{subword}}}. A strong Dyck word consists of exactly one subword.
\end{definition}

\subsection{Transformation}
A series of transformations can be applied to Dyck word. It can consist in adding, removing\cite{chistikov_re-pairing_2020}, moving\cite{baril_phagocyte_2006}, or flipping brackets\cite{barnabei_permutations_2016}. For instance, the concatenation mentioned in Proposition \ref{prop:weakdyck} is a transformation on Dyck words. In this section, we will explore a transformation called \textit{expulsion}.

\begin{definition}[Expulsion]
Let $S$ be a non-empty Dyck word of length $2n$. An {\normalfont{expulsion}} consists in removing one opening and one closing bracket on its right in $S$ to generate a new string $S'$. 
\end{definition}
\noindent After an expulsion, the resulting string may or may not be a Dyck word. In Example \ref{ex:dyck_path_expulsion} the expulsion results in a Dyck word.
\begin{example}\label{ex:dyck_path_expulsion}
Consider the $\bm{(\,(\,)\,(\,)\,)\,(\,)}$ Dyck word and the corresponding list of depths $(0,1,2,1,2,1,0,1,0)$. If the transformation consists in expulsing the second and fifth bracket, it can be described as follows:
\begin{equation*}
\stackrel{1}{\bm{(}}\;
\stackrel{2}{\bm{(}}\;
\stackrel{3}{\bm{)}}\;
\stackrel{4}{\bm{(}}\;
\stackrel{5}{\bm{)}}\;
\stackrel{6}{\bm{)}}\;
\stackrel{7}{\bm{(}}\;
\stackrel{8}{\bm{)}}\;
\quad\stackrel{\text{expulsion}}{\longrightarrow} \quad
\stackrel{1}{\bm{(}}\;
\;
\stackrel{3}{\bm{)}}\;
\stackrel{4}{\bm{(}}\;
\;
\stackrel{6}{\bm{)}}\;
\stackrel{7}{\bm{(}}\;
\stackrel{8}{\bm{)}}.
\end{equation*}
The resulting string, i.e., $\bm{(\;)\;(\;)\;(\;)}$, is a Dyck word described by the $(0,1,0,1,0,1,0)$ list of depths.
\end{example}
\noindent On the other hand, Example \ref{ex:dyck_path_expulsion_nondyck} provides an illustration of how an expulsion can result in a string that is not a Dyck word.
\begin{example}\label{ex:dyck_path_expulsion_nondyck}
Consider the $\bm{(\,(\,)\,(\,)\,)\,(\,)}$ Dyck word and the corresponding list of depths $(0,1,2,1,2,1,0,1,0)$. If the transformation consists in expulsing the second and eighth bracket, it can be described as follows:
\begin{equation*}
\stackrel{1}{\bm{(}}\;
\stackrel{2}{\bm{(}}\;
\stackrel{3}{\bm{)}}\;
\stackrel{4}{\bm{(}}\;
\stackrel{5}{\bm{)}}\;
\stackrel{6}{\bm{)}}\;
\stackrel{7}{\bm{(}}\;
\stackrel{8}{\bm{)}}\;
\quad\stackrel{\text{expulsion}}{\longrightarrow} \quad
\stackrel{1}{\bm{(}}\;
\;
\stackrel{3}{\bm{)}}\;
\stackrel{4}{\bm{(}}\;
\stackrel{5}{\bm{)}}\;
\stackrel{6}{\bm{)}}\;
\stackrel{7}{\bm{(}}\;
\end{equation*}
The resulting string, i.e., $\bm{(\;)\;(\;)\;)\;(}$, is not a Dyck word according to criterion 1 in {\normalfont Definition \ref{def:dyck_word}}.
\end{example}

\begin{proposition}\label{prop:action_association}
Let $S$ be a non-empty Dyck word of length $2n$. Let $(d_k)_{0 \leq k \leq 2n}$ be the list of depths of $S$. Let the expulsed opening bracket be the $n_1^{\text{th}}$ and the closing bracket be the $n_2^{\text{th}}$ (with, of course, $n_1 < n_2$). Let $(d'_k)_{0 \leq k \leq 2n-2}$ be list of depths of $S'$. For $k$ ranging from zero to $2n-2$, one can encounter the three following cases:
\begin{equation*}
\begin{cases}
	d'_k = d_k			&	\quad 0   \leq k < n_1,	\\
	d'_k = d_{k+1} - 1	&	\quad n_1 \leq k < n_2,		\\
	d'_k = d_{k+2}		&	\quad n_2 \leq k \leq 2n-2,	\\
\end{cases}
\end{equation*}  
\end{proposition}

\begin{proposition}\label{prop:expulsion}
In a non-empty Dyck word, expulsing an opening bracket with any closing bracket to its right within the same subword results in a Dyck word.
\end{proposition}
\begin{proof}
According to Proposition \ref{prop:action_association}, and since subwords satisfy $d_k \geq 1$ except at their boundaries, the lowest depth value one can meet within the subword after the expulsion is performed remains zero. Hence, in agreement with propositions \ref{prop:depthdycknonnegative} and \ref{prop:action_association} the result of the transformation remains a Dyck word.
\end{proof}

\noindent In line with what has just been written, notice that in Example \ref{ex:dyck_path_expulsion} one expulses two brackets within the same subword, leading to a Dyck word, while in Example \ref{ex:dyck_path_expulsion_nondyck} one expulses two brackets which do not belong to the same subword, leading to a string that is not a Dyck word.

\begin{proposition}\label{prop:denombrement-dyck}
The number of ways to iteratively expulse all the brackets pairs of a non-empty Dyck word starting at each step from the last opening bracket, in such a way that after each expulsion the resulting string is a Dyck word, is equal to the product of its opening depths.
\end{proposition}

\begin{proof}
The last opening bracket in a $2n$--long Dyck word has exactly $d^\uparrow_n$ closing brackets on its right. Indeed, the depth after this bracket is $d^\uparrow_n$ and must return to zero at the end of the word. According to Proposition \ref{prop:expulsion}, expulsing the last opening bracket and any of these $d^\uparrow_n$ closing brackets lead to a Dyck word. 

According to Proposition \ref{prop:action_association}, solely depths after the last opening bracket will change upon an expulsion transformation while all the depths (hence, the opening depths) before the last opening bracket will remain the same.

After the expulsion has been performed, the last opening bracket has $d^\uparrow_{n-1}$ closing brackets on its right. After the $k^\text{th}$ expulsion, the last opening bracket has $d^\uparrow_{n-k}$ choices on its right. In conclusion, the total number of possibilities is given by the product of all the opening depths.
\end{proof}

\noindent This proposition will be particularly relevant in the second part of our contribution, where we will see how we can connect Dyck language with a corollary of the time-independent fermionic Wick theorem.

\begin{example}
Let $S$ be the $\bm{(\;(\;)\;)\;(\;)}$ Dyck word of length six. Its depth list is {\normalfont$(0,1,2,1,0,1,0)$}. Its opening depth list is {\normalfont$(1,2,1)$}. Therefore, the number of ways to expulse all the brackets pairs from $S$ according to {\normalfont Proposition \ref{prop:denombrement-dyck}} is the product of elements from the opening depth list, i.e., $1\cdot 2\cdot 1 = 2$. The two possible expulsion routes in agreement with the procedure depicted in {\normalfont Proposition \ref{prop:denombrement-dyck}} are reproduced in {\normalfont Figure \ref{fig:dyckexpulsiontotale}}.

\begin{figure}[h!]\centering
\begin{tikzpicture}
\node (start) at (.5,1.5) {Start};
\draw[->, red] (start) to (.5,.4) ;

\node (step_0)      at (0,0)  
  {$\stackrel{1}{\bm{(}}\; 
    \stackrel{2}{\bm{(}}\; 
    \stackrel{3}{\bm{)}}\; 
    \stackrel{4}{\bm{)}}\; 
    \tcr{\stackrel{5}{\bm{(}}}\; 
    \tcb{\stackrel{6}{\bm{)}}} $};

\node (ex1) at (1.5,-1)    {$(\tcr{5},\tcb{6})$};
\node (step_1)      at (0,-2)  
  {$\stackrel{1}{\bm{(}}\; 
    \tcr{\stackrel{2}{\bm{(}}}\; 
    \tcb{\stackrel{3}{\bm{)}}\; 
    \stackrel{4}{\bm{)}}} $};

\node (s2) at (0,-3.5) {};

\node (ex2)       at (-1.5,-3.5)    {$(\tcr{2},\tcb{3})$};
\node (step_2)    at (-2.5,-4.5)  
  {$\tcr{\stackrel{1}{\bm{(}}}\; 
    \tcb{\stackrel{4}{\bm{)}}} $};
\node (ex2')      at (1.5,-3.5)    {$(\tcr{2},\tcb{4})$};
\node (step_2')   at (2.5,-4.5)  
  {$\tcr{\stackrel{1}{\bm{(}}}\;  
    \tcb{\stackrel{3}{\bm{)}}} $};

\node (ex3)       at (-3.5,-5.4)    {$(\tcr{1},\tcb{4})$};
\node (ex3')      at (3.5,-5.4)     {$(\tcr{1},\tcb{3})$};
\node (step_3)    at (-2.5,-6)  {End};
\node (step_3')   at (2.5,-6)   {End};

\draw[->] (step_0) to (step_1) node[below=-1cm, left] {Step 1} ;
\draw[-, anchor=base] (step_1) --  (s2.base);
\draw[->,anchor=base] (s2.base) -- (step_2)  node[midway, below=.23cm] {Step 2};
\draw[->,anchor=base] (s2.base) -- (step_2') node[midway, below=.23cm] {Step 2'};

\draw[->] (step_2)  -- (step_3);
\draw[->] (step_2') -- (step_3');

\draw[->] (step_0.south)  to [bend right=45]  (ex1.west);
\draw[->] (s2.base)       to [bend left=25]   (ex2.east);
\draw[->] (s2.base)       to [bend right=25]  (ex2'.west);
\draw[->] (step_2.south)  to [bend left=25]   (ex3.east);
\draw[->] (step_2'.south) to [bend right=25]  (ex3'.west);
\end{tikzpicture}

\caption{Illustration of a total expulsion of brackets the $\bm{(\;(\;)\;)\;(\;)}$ Dyck word according to the procedure depicted in Proposition \ref{prop:denombrement-dyck}. In Step 1 we expulse the last opening bracket together with the only closing one on its right. In the following step the now-last opening bracket has two possible choices for pairing with a closing one, leading to Step 2 and to Step 2'. The last pairings are trivial — see Remark \ref{rem:emptyword}. The two possible expulsion routes are therefore $((5-6),(2-3),(1-4))$ and $((5-6),(2-4),(1-3))$.}
\label{fig:dyckexpulsiontotale}
\end{figure}
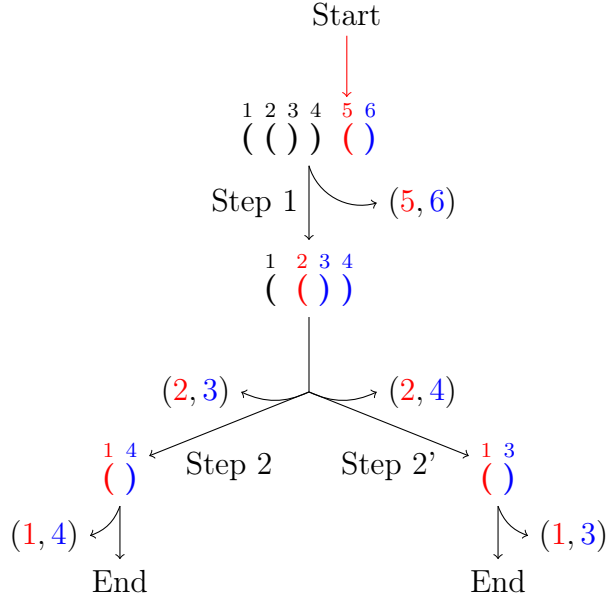
\end{example}

\section{Connecting second quantization with Dyck language}

Comparing Proposition \ref{prop:criteria_nullity_EVC_2nd} and Definition \ref{def:dyck_word}, we immediately notice a one-to-one correspondence between criteria defining a Dyck word and criteria ensuring the expectation value nullity of a chain of second quantization operators relatively to the physical vacuum. Indeed, both are defined on two criteria: a syntactic criterion and a criterion related to what is contained in the chain/string. This observation strongly invites us to suggest a ``translation'' of the operators chains using the $\mathcal{L}_1$ alphabet.

\newpage \subsection{Nullity criteria relatively to the physical vacuum}
When expectation values are evaluated relatively to the physical vacuum, one set of criteria appears.

\begin{definition}[$\mathcal{L}_1$ translation]\label{def:L1}
$\mathcal{L}_1$-translating a chain of individual second quantization operators consists in replacing each annihilation operator by a opening bracket {\normalfont ``}$\bm{(}${\normalfont ''}, and each creation operator by a closing bracket {\normalfont ``}$\bm{)}${\normalfont ''}:
\begin{align*}
\forall r \in \llbracket 1, L \rrbracket, \quad &\hat{a}_r \longrightarrow_{\mathcal{L}_1} \bm{(},\\ &\hat{a}^\dag_r \longrightarrow_{\mathcal{L}_1} \;\bm{)}.
\end{align*}
\end{definition}

\begin{example}\label{ex:translation_physical_vacuum_L1}
Let $(p,q,r,s)$ be a 4-tuple of integers, all belonging to $\llbracket 1,L\rrbracket$. Consider the three following chains of operators: $\hat{a}_p\hat{a}_q\hat{a}^\dag_s\hat{a}^\dag_r$, $\hat{a}_p\hat{a}_q\hat{a}^\dag_s\hat{a}^\dag_r$, and $\hat{a}_p^\dag\hat{a}_q\hat{a}^\dag_s\hat{a}_r$. Their respective translation in $\mathcal{L}_1$ are:
\begin{equation*}
\hat{a}_p\hat{a}_q\hat{a}^\dag_s\hat{a}^\dag_r \longrightarrow_{\mathcal{L}_1} \bm{(())},
\quad
\hat{a}_p\hat{a}_q\hat{a}_s\hat{a}^\dag_r \longrightarrow_{\mathcal{L}_1} \bm{((()},
\quad
\hat{a}_p^\dag\hat{a}_q\hat{a}^\dag_s\hat{a}_r \longrightarrow_{\mathcal{L}_1}\, \bm{)()(}.
\end{equation*}
\end{example}

\begin{proposition}\label{prop:nullity_L1}%[Nullity criterion in physical vacuum]
A chain of second quantization operators has a zero expectation value relatively to the physical vacuum if its $\mathcal{L}_1$--translation results into a string of bracket(s) that is not a Dyck word.
\end{proposition}
\begin{proof} To any chain of second quantization operators corresponds one and only one string in $\mathcal{L}_1$. On the other hand, there is a one-to-one correspondence between the criteria defining Dyck words and criteria establishing the nullity of expectation values of second quantization operators chains relatively to the physical vacuum. 
\end{proof}

\begin{example}
Consider the three chains of operators defined in {\normalfont Example \ref{ex:translation_physical_vacuum_L1}}. Once translated in $\mathcal{L}_1$, the last two chains do not lead to a Dyck word. Therefore, their expectation value relatively to the physical vacuum is both equal to zero in both cases. On the other hand, the $\mathcal{L}_1$ translation of the first chain is a Dyck word. Consequently, the expectation value of that chain relatively to the physical vacuum can be non-zero.
\end{example}

\noindent The translation is not mandatory to check the nullity. Indeed, one could directly use Proposition \ref{prop:criteria_nullity_EVC_2nd}. Yet, the Dyck language has the advantage to rely on the ability to identify a well-nested string of brackets more easily.

\subsection{Nullity criteria relatively to the Fermi vacuum}
\subsubsection{Translating chains of individual \textit{o}- and \textit{v}-operators}

\noindent In the Fermi vacuum, two sets of mirrored rules appears: One for \textit{o-}operators, and one for \textit{v-}operators. It invites us to use two different translations, one for occupied and one for virtual spaces, respectively. For this seek, one first needs to build two strings of bracket.

Let $\mathcal{L}_2 \coloneqq \{\, \bm{)}, \bm{(}, \bm{]},\bm{[}\,\}$ be the alphabet consisting in the four following symbols: the opening round bracket ``$\bm{(}$'', the closing round bracket ``$\bm{)}$'', the opening square bracket ``$\bm{[}$'', and the closing square bracket ``$\bm{]}$''. 

\begin{definition}[$\mathcal{L}_2$ translation]\label{def:L2}
$\mathcal{L}_2$-translating a chain of individual second quantization operators consists in replacing each \textit{o-}annihilation operator by a closing square bracket {\normalfont ``}$\bm{]}${\normalfont ''}, and each \textit{o-}creation operator by an opening square bracket {\normalfont ``}$\bm{[}${\normalfont ''}, i.e., using the
\begin{align*}
\forall i \in \llbracket 1, N \rrbracket, \quad &\hat{o}_i \longrightarrow_{\mathcal{L}_2} \;\bm{]},\\ &\hat{o}^\dag_i \longrightarrow_{\mathcal{L}_2} \bm{[}
\end{align*}
translation rules, and in replacing each \textit{v-}annihilation operator by a opening round bracket {\normalfont ``}$\bm{(}${\normalfont ''}, and each \textit{v-}creation operator by a closing round bracket {\normalfont ``}$\bm{)}${\normalfont ''}, i.e., using the
\begin{align*}
\forall a \in \llbracket N+1, L \rrbracket, \quad 
&\hat{v}_a      \longrightarrow_{\mathcal{L}_2}   \bm{(},\\ 
&\hat{v}^\dag_a \longrightarrow_{\mathcal{L}_2}\; \bm{)}
\end{align*}
translation rules.
\end{definition}

\begin{definition}[Splitting]\label{def:splitting}
Let $S$ be a string written using elements of $\mathcal{L}_2$. $S$ can be split into two strings, $\text{S}_o$ and $\text{S}_v$, separated by a $\&$ symbol. $S_o$ is composed of all the \textit{o-}operators of $S$ met in the same order as in $S$ while $S_v$ is composed of all the \textit{v-}operators of $S$ met in the same order as in $S$:
\begin{equation*}
    S \longrightarrow_{\mathcal{S}} S_o\, \&\, S_v
\end{equation*}
\end{definition}
\noindent For the sake of clarity, in Definition \ref{def:splitting}, we use the ``$\&$'' symbol as a meta-syntactic symbol for separating $S_o$ and $S_v$. This notation is used rather than, for instance, $(S_o , S_v)$, in order to avoid a mixing of symbols — in this case, round brackets — used for both syntactic and meta-syntactic purpose in the same expression.

\begin{example}\label{ex:split-translation-dyck}
    Let $(i,j,k,l)$ be a 4-tuple of integers, all belonging to $\llbracket 1,N\rrbracket$. Let $(a,b)$ be a couple of integers, both belonging to $\llbracket N+1,L\rrbracket$. Consider the $\hat{o}_i^\dag\hat{v}_a\hat{o}_j^\dag\hat{o}_k\hat{v}_b^\dag\hat{o}_l$ chain. Its $\mathcal{L}_2$--translation, followed by the splitting of the round and square bracket sequences is illustrated in Figure \ref{fig:ex_translation_L2}.
\end{example}
\begin{figure}[h!]
\begin{equation*}
 \tcb{\hat{o}^\dag_i}
 \tcr{\hat{v}_a}
 \tcb{\hat{o}^\dag_j\hat{o}_k}
 \tcr{\hat{v}^\dag_b}
 \tcb{\hat{o}_l}
 \longrightarrow_{\mathcal{L}_2} 
 \tcb{\bm{[}}\; 
 \tcr{\bm{(}}\; 
 \tcb{\bm{[\;]}}\; 
 \tcr{\bm{)}}\; 
 \tcb{\bm{]}}
 \longrightarrow_{\mathcal{S}}
 \tcb{\bm{[\;[\;]\;]}} \;\&\;  \tcr{\bm{(\;)}} 
\end{equation*}
\caption{Illustration of the translation $(\longrightarrow _{\mathcal{L}_2})$ of $\hat{o}_i^\dag\hat{v}_a\hat{o}_j^\dag\hat{o}_k\hat{v}_b^\dag\hat{o}_l$ in the $\mathcal{L}_2$ alphabet performed according to Definition \ref{def:L2}, followed by its splitting $(\longrightarrow _{\mathcal{S}})$ performed according to Definition \ref{def:splitting}.}
\label{fig:ex_translation_L2}
\end{figure}
\begin{definition}[$\mathcal{L}_2$--split translation]\label{def:split-translation}
    The consecutive application of the $\mathcal{L}_2$ translation and the splitting operation will be termed ``$\mathcal{L}_2$--split translation''.
\end{definition}
\begin{example}\label{ex:split-translation-non-dyck}
Let $(i,j,k)$ be a 3-tuple of integers, all belonging to $\llbracket 1,N\rrbracket$. Let $(a,b,c)$ be a 3-tuple of integers, all belonging to $\llbracket N+1,L\rrbracket$. Consider the $\hat{o}_i^\dag\hat{v}_a\hat{o}_j^\dag\hat{v}_b^\dag\hat{v}_c^\dag\hat{o}_k$ chain. Its $\mathcal{L}_2$--split translation is illustrated in Figure \ref{fig:ex_translation_L2_nondyck}.
\begin{figure}[h!]
\begin{equation*}
 \tcb{\hat{o}_i^\dag}
 \tcr{\hat{v}_a}
 \tcb{\hat{o}_j^\dag}
 \tcr{\hat{v}_b^\dag}
 \tcr{\hat{v}_c^\dag}
 \tcb{\hat{o}_k}
 \longrightarrow _{\mathcal{L}_2,\mathcal{S}}
 \tcb{\bm{[\;[\;]}} \;\&\;  \tcr{\bm{(\;)\;)}} 
\end{equation*}
\caption{Illustration of the $\mathcal{L}_2$--split translation $(\longrightarrow_{\mathcal{L}_2,\mathcal{S}})$ of $\hat{o}_i^\dag\hat{v}_a\hat{o}_j^\dag\hat{v}_b^\dag\hat{v}_c^\dag\hat{o}_k$.}
\label{fig:ex_translation_L2_nondyck}
\end{figure}
\end{example}

\begin{proposition}\label{prop:nullity_translation_Fermi}
    A chain of second quantization operators has a zero expectation value relatively to the Fermi vacuum if at least one of the two strings of its $\mathcal{L}_2$-split translation is not a Dyck word.
\end{proposition}
\begin{proof}
    First, one should notice that, according to Remark \ref{rem:round-square-dyck}, well-nested strings of square brackets are Dyck words, and strings of square brackets that are not well nested are not Dyck words.

    To any chain of second quantization operators corresponds one and only one pair of strings after $\mathcal{L}_2$--split translation. On the other hand, there is a one-to-one correspondence between the criteria defining Dyck words — see Definition \ref{def:dyck_word} — and criteria establishing the nullity of expectation values of \textit{o-}operators and \textit{v-}operators relatively to the Fermi vacuum — see \ref{prop:criteria_nullity_EVC_fermi_2nd}. 
\end{proof}

\begin{example}
    Consider the $\hat{o}_i^\dag\hat{v}_a\hat{o}_j^\dag\hat{o}_k\hat{v}_b^\dag\hat{o}_l$ chain introduced in {\normalfont Example \ref{ex:split-translation-dyck}}. Both strings obtained \textit{via} the $\mathcal{L}_2$-split translation are Dyck words. Therefore, according to {\normalfont Proposition \ref{prop:nullity_translation_Fermi}} expectation value of the chain relatively to the Fermi vacuum can be either zero or non-zero. On the other hand, the $\hat{o}_i^\dag\hat{v}_a\hat{o}_j^\dag\hat{v}_b^\dag\hat{v}_c^\dag\hat{o}_k$ chain introduced in {\normalfont Example \ref{ex:split-translation-non-dyck}}, after $\mathcal{L}_2$-split translation, does not result into a pair of Dyck words. Therefore, according to {\normalfont Proposition \ref{prop:nullity_translation_Fermi}} its expectation value relatively to the Fermi vacuum is equal to zero.
\end{example}
\noindent We have introduced the $\mathcal{L}_2$ alphabet in order to supplement $\mathcal{L}_1$ so that one can handle two sets of operators defined relatively to two separate parts of $\mathcal{B}$. One can extend this supplementation as follows:
\begin{equation*}
\forall i \in \llbracket 3,L\rrbracket, \, \mathcal{L}_i \coloneqq \bigcup_{j\in\llbracket 1,i\rrbracket} \left\{ \,\bm{)}_j, \bm{(}_j \,\right\}
\end{equation*}
together with the corresponding splitting: Let $\hat{C}$ be a chain of second quantization operators. Its $\mathcal{L}_i$-split translation will have the following structure:
\begin{equation*}
\hat{C} \longrightarrow_{\mathcal{L}_{i},\mathcal{S}} S_1\, \& \cdots \& \,S_i.
\end{equation*}
This extension may be useful in the framework of complete active space self-consistent field (CAS-SCF) method in which the working space is split into three parts.

\subsubsection{Excitation and deexcitation operators}

\noindent In this paragraph we will be concerned exclusively with chains solely composed of excitation and deexcitation operators introduced in \eqref{eq:ExcOp} and \eqref{eq:DeexcOp}, respectively.

\begin{proposition}\label{prop:nullity_L2}
    Let $\hat{C}$ be a chain of (de)excitation operators. The two strings — $S_o$ and $S_v$ — of the $\mathcal{L}_2$-split translation of $\hat{C}$ have the same sequence of opening and closing brackets.
\end{proposition}
\begin{proof}
    According to definitions \ref{def:L2} to \ref{def:split-translation}, the $\mathcal{L}_2$-split translation of excitation and deexcitation operators reads 
    \begin{align*}
        \forall (i,a) \in \llbracket 1, N \rrbracket\times \llbracket N+1, L \rrbracket, \quad 
        &\hat{D}_a^i=\hat{o}_i^\dag\hat{v}_a \longrightarrow_{ \mathcal{L}_2} \bm{[\;(} \;\longrightarrow_{ \mathcal{S}}   \bm{[}\;\&\;\bm{(},\\ 
        &\hat{E}_i^a=\hat{v}_a^\dag\hat{o}_i\longrightarrow_{ \mathcal{L}_2} \,\;\bm{)\;]}  \longrightarrow_{ \mathcal{S}}\; \bm{)}\;\&\;\bm{]}.
    \end{align*}
In a chain composed of such operators, the $\mathcal{L}_2$--translation of each deexcitation (respectively, excitation) operator pairs an opening (respectively, closing) square bracket with an opening (respectively, closing) round bracket directly to its right (respectively, left).

After the splitting is performed, in both cases if one encounters an opening (respectively, closing) bracket at the left of ``$\&$'', an opening (respectively, closing) bracket is also encountered at the right of ``$\&$''.

Thus, if a string is composed solely of excitation and deexcitation operators, the two strings of its $\mathcal{L}_2$-split translation will have the same sequence of opening and closing brackets.
\end{proof}
\begin{example}\label{ex:DEL2S}
Let $(i,j)$ be a couple of integers, both belonging to $\llbracket 1,N\rrbracket$. Let $(a,b)$ be a couple of integers, both belonging to $\llbracket N+1,L\rrbracket$. Consider the $\hat{D}_a^i\hat{E}_j^b$ and $\hat{E}_j^b\hat{D}_a^i$ chains. Their $\mathcal{L}_2$-split translation are
\begin{align*}
 \hat{D}_a^i
 \hat{E}_j^b
&=
 \hat{o}^\dag_i
 \hat{v}_a
 \hat{v}^\dag_b
 \hat{o}_j
\longrightarrow_{\mathcal{L}_2,\mathcal{S}} 
 \bm{[\,]} \;\&\;  \bm{()},
\\
 \hat{E}_j^b
 \hat{D}_a^i
&=
 \hat{v}^\dag_b
 \hat{o}_j
 \hat{o}^\dag_i
 \hat{v}_a
\longrightarrow_{\mathcal{L}_2,\mathcal{S}} \,\;
 \bm{][} \;\&\; \bm{)(}.
\end{align*}
\end{example}
\noindent A direct consequence of Proposition \ref{prop:nullity_L2} is that if one of the two strings obtained by $\mathcal{L}_2$--split translation is (respectively, is not) a Dyck word, the other one is (respectively, is not) a Dyck word. This redundancy of information in the $\mathcal{L}_2$--split translation for chains of such operators strongly invites us to translate the chain in $\mathcal{L}_1$ instead of $\mathcal{L}_2$.

\begin{definition}[Fermi-$\mathcal{L}_1$ translation]\label{def:F_L1}
Fermi-$\mathcal{L}_1$ translating a chain of (de)excitation operators consists in replacing each deexcitation operator by a opening bracket {\normalfont ``}$\bm{(}${\normalfont ''}, and each excitation operator by a closing bracket {\normalfont ``}$\bm{)}${\normalfont ''}:
\begin{align*}
\forall (i,a) \in \llbracket 1, N \rrbracket\times \llbracket N+1, L \rrbracket, \quad 
&\hat{D}_a^i \stackrel{\mathcal{F}}{\longrightarrow}_{\mathcal{L}_1}   \bm{(},\\ 
&\hat{E}_i^a \stackrel{\mathcal{F}}{\longrightarrow}_{\mathcal{L}_1} \;\bm{)}.
\end{align*}
\end{definition}
\noindent We named this translation ``Fermi-$\mathcal{L}_1$ translation'' in order not to confuse this operation with the $\mathcal{L}_1$ translation introduced in Definition \ref{def:L1}. Whilst the arguments of a Fermi-$\mathcal{L}_1$ translation are (de)excitation operators, those of the $\mathcal{L}_1$ translation are creation and annihilation second quantization operators. The difference is illustrated in Example \ref{ex:L1FermiL1} where the same chain is translated into two different strings in the $\mathcal{L}_1$ alphabet when seen as chains of (de)excitation operators or as chains of creation and annihilation operators.

\begin{example}\label{ex:L1FermiL1}
Let $(i,j)$ be a couple of integers, both belonging to $\llbracket 1,N\rrbracket$. Let $(a,b)$ be a couple of integers, both belonging to $\llbracket N+1,L\rrbracket$. Consider the $\hat{D}_a^i\hat{E}_j^b$ and $\hat{E}_j^b\hat{D}_a^i$ chains. Their $\mathcal{L}_1$ translation and Fermi-$\mathcal{L}_1$ translation are
\begin{align*}
\bm{)\;(\;)\;(}
\; {}_{\mathcal{L}_1}\!\!\longleftarrow
 \hat{o}^\dag_i
 \hat{v}_a
 \hat{v}^\dag_b
 \hat{o}_j
 &=
 \hat{D}_a^i
 \hat{E}_j^b
 \stackrel{\mathcal{F}}{\longrightarrow}_{\mathcal{L}_1}
 \bm{(\;)},
\\
 \bm{)\;(\;)\;(}
 \;{}_{\mathcal{L}_1}\!\!\longleftarrow
 \hat{v}^\dag_b
 \hat{o}_j
 \hat{o}^\dag_i
 \hat{v}_a
& =
 \hat{E}_j^b
 \hat{D}_a^i
 \stackrel{\mathcal{F}}{\longrightarrow}_{\mathcal{L}_1}\,\;
 \bm{)(}.
\end{align*}
We can notice that the $\mathcal{L}_1$ translation of both chains is the same while their Fermi-$\mathcal{L}_1$ translations are different.
\end{example}
\begin{proposition}\label{prop:nullity_F_L1}
The expectation value of a chain of (de)excitation operator(s) relatively to the Fermi vacuum is zero if its Fermi-$\mathcal{L}_1$ translation is not a Dyck word.
\end{proposition}

\begin{proof}
Let $\hat{C}$ be a chain of (de)excitation operators. Let $S$ be the Fermi-$\mathcal{L}_1$ translation of $\hat{C}$ and $S_o\;\&\;S_v$ be the $\mathcal{L}_2$-split translation of $\hat{C}$.
According to Proposition \ref{prop:nullity_L2} and Definition \ref{def:F_L1}, $S$ and $S_v$ are the same. According to Proposition \ref{prop:nullity_L2}, $S_o$ and $S_v$ have the same sequence of opening and closing brackets. Thus, if $S$ is not a Dyck word, according to Proposition \ref{prop:nullity_translation_Fermi} the expectation value of $\hat{C}$ relatively to the Fermi vacuum is equal to zero.
\end{proof}
\begin{example}\label{ex:L1FermiL1vsDEL2S}
Consider the two chains in {\normalfont Example \ref{ex:L1FermiL1}}. According to {\normalfont Proposition \ref{prop:nullity_F_L1}}, we see that the first of the two chains is susceptible of having a non-zero expectation value relatively to the Fermi vacuum while one can readily assert the nullity of expectation value of the second chain relatively to the Fermi vacuum. 
\end{example}
\begin{remark}
One can see that the conclusions drawn in {\normalfont Example \ref{ex:L1FermiL1vsDEL2S}} are identical to those one could draw from using {\normalfont Proposition \ref{prop:nullity_translation_Fermi}} for comparing the two chains using their $\mathcal{L}_2$--split translation given in {\normalfont Example \ref{ex:DEL2S}}. However, using Fermi-$\mathcal{L}_1$ translation for going from the chain to the expectation value implies writing only one string of brackets, which should be compared with the $\mathcal{L}_2$--split translation — that implies two strings.
\end{remark}
{\begin{remark}
{Deexcitation operators are to the Fermi vacuum what annihilation operators are to the physical vacuum. Indeed, while on one hand applying an annihilation operator to the physical vacuum results in the null state, on the other hand applying a deexcitation operator to the Fermi vacuum results in $\ket{0_N}$. Similarly, excitation operators are to the Fermi vacuum what creation operators are to the physical vacuum.}
\end{remark}}{}

\subsubsection{Arbitrary creation-annihilation operator pairs}

\noindent Until now the subspace in which target spin-orbitals belong when pointed by second quantization operators was specified. Let $(r,s)$ be a couple of integers, both belonging to $\llbracket 1,L\rrbracket$. An arbitrary creation-annihilation pair $\hat{a}^\dag_r\hat{a}_s$ — in that specific order — is not always an excitation or deexcitation operator. In fact, four possibilities arise depending on the occupation number of the spin-orbital ``$r$'' and ``$s$'' are pointing at. The $\mathcal{L}_2$ translation of $\hat{a}^\dag_r\hat{a}_s$ therefore depends on the value of $r$ and $s$, as illustrated in Table \ref{tab:adaga_occ_vir}.

\begin{table}[h!]
\def\arraystretch{1.5}
\begin{center}
\begin{tabular}{ccc}
\hline 
\diagbox[width=\dimexpr \textwidth/8+2\tabcolsep\relax, height=1cm]{$r$}{$s$}   & Occupied    & Virtual     \\ 
\hline 
Occupied & $\bm{[\;]}$  & $\bm{[\;(}$ \\ 
Virtual  & $\bm{)\;]}$  & $\bm{)\;(}$ \\ 
\hline 
\end{tabular} 
\end{center}
\caption{$\mathcal{L}_2$ translation of $\hat{a}^\dag_r\hat{a}_s$ according to the type of spinorbital $r$ and $s$ are pointing at.}
\label{tab:adaga_occ_vir}
\end{table}

\noindent In order to establish sufficient conditions for the nullity of the expectation value of a chain containing solely (de)excitation operators and arbitrary creation and annihilation operator pairs through an $\mathcal{L}_2$ translation, four possibilities have to be explored. To avoid four translations in $\mathcal{L}_2$ alphabet the $\mathcal{L}_1$ alphabet is now extended in order to take into account this new pattern: Let $\mathcal{L}_{1,1} \coloneqq \mathcal{L}_1 \cup \{\bm{-}\}$ be the alphabet consisting in the opening round bracket ``$\bm{(}$'', the closing round bracket ``$\bm{)}$'', and the dash symbol ``$\bm{-}$''.

\begin{definition}[Fermi-$\mathcal{L}_{1,1}$ translation]\label{def:F_L1-extend}
Fermi-$\mathcal{L}_{1,1}$ translating a chain containing solely (de)excitation operators and one arbitrary creation-annihilation operator pair consists in replacing each deexcitation operator by a opening bracket {\normalfont ``}$\bm{(}${\normalfont ''}, each excitation operator by a closing bracket {\normalfont ``}$\bm{)}${\normalfont ''}
\begin{align*}
\forall (i,a) \in \llbracket 1, N \rrbracket\times \llbracket N+1, L \rrbracket, \quad 
\hat{D}_a^i              &\stackrel{\mathcal{F}}{\longrightarrow}_{\mathcal{L}_{1,1}}   \bm{(},\\ 
\hat{E}_i^a              &\stackrel{\mathcal{F}}{\longrightarrow}_{\mathcal{L}_{1,1}} \,\;\bm{)},
\end{align*}
and the arbitrary creation-annihilation pair by a dash {\normalfont ``}$\bm{-}${\normalfont ''}, e.g.,
\begin{align*}
\hat{a}_r^\dag\hat{a}_s  &\stackrel{\mathcal{F}}{\longrightarrow}_{\mathcal{L}_{1,1}}   \bm{-},
\end{align*}
where ``$r$'' and ``$s$'' are two non-definite indices about which we know that each can be valued from $1$ to $L$.
\end{definition}

\begin{example}\label{ex:translation_Fermi_L-1-1}
Let $(i,j,k)$ be a 3-tuple of integers, all belonging to $\llbracket 1,N\rrbracket$. 
Let $(a,b,c)$ be a 3-tuple of integers, all belonging to $\llbracket N+1,L\rrbracket$. 
Let $\hat{a}^\dag_r$ and $\hat{a}_s$ be two second quantization operators, each pointing at a spin-orbital with non-definite occupation number relatively to the Fermi vacuum.
Consider the three following chains of second quantization operators: $\hat{D}_a^i\hat{a}^\dag_r\hat{a}_s\hat{E}_j^b$, $\hat{D}_a^i\hat{a}^\dag_r\hat{a}_s\hat{E}_j^b\hat{E}_k^c $, and $\hat{D}_a^i\hat{D}_b^j\hat{a}^\dag_r\hat{a}_s$. Their Fermi-$\mathcal{L}_{1,1}$ translation reads 
\begin{equation*}
\hat{D}_a^i\hat{a}^\dag_r\hat{a}_s\hat{E}_j^b
\stackrel{\mathcal{F}}{\longrightarrow}_{\mathcal{L}_{1,1}}\bm{(-)}, \quad
\hat{D}_a^i\hat{a}^\dag_r\hat{a}_s\hat{E}_j^b\hat{E}_k^c
\stackrel{\mathcal{F}}{\longrightarrow}_{\mathcal{L}_{1,1}} \bm{(-))}, \quad
\hat{D}_a^i\hat{D}_b^j\hat{a}^\dag_r\hat{a}_s
\stackrel{\mathcal{F}}{\longrightarrow}_{\mathcal{L}_{1,1}}\bm{((-}.
\end{equation*}
\end{example}

\begin{proposition}\label{prop:nullity_F_L1_extend}
Let $\hat{C}$ be a chain containing (de)excitation operators and one arbitrary creation-annihilation operator pair. Let $S$ be its Fermi-$\mathcal{L}_{1,1}$ translation.
The expectation value of $\hat{C}$ relatively to the Fermi vacuum is zero if none of the following criteria is met:    
\begin{itemize}
    \item[P1] The $S$ string is a Dyck word after replacing the dash by an opening round bracket.
    \item[P2] The $S$ string is a Dyck word after replacing the dash by an closing round bracket.
    \item[P3] The $S$ string is a Dyck word after removing the dash.
\end{itemize}
\end{proposition}
\begin{proof}
Let $\hat{a}^\dag_r$ and $\hat{a}_s$ be two second quantization operators, each pointing at a spin-orbital with non-definite occupation number relatively to the Fermi vacuum. Let $S^E$ (respectively, $S^D$) be the string obtained by replacing the dash by a closing (respectively, opening) bracket in $S$. Let $S_o\;\&\;S_v$ be the $\mathcal{L}_2$-translation of $\hat{C}$. Let and $S'_o\;\&\;S'_v$ be the $\mathcal{L}_2$-translation of $\hat{C}$ stripped of the $\hat{a}_r^\dag\hat{a}_s$ pair.

The negation of the $P1$, $P2$, and $P3$ propositions can be written as:
\begin{equation}\label{eq:negP1P2P3}
(\neg \textit{P1}) \Longleftrightarrow {\bf C1}, \quad 
(\neg \textit{P2}) \Longleftrightarrow {\bf C2}, \quad 
(\neg \textit{P3}) \Longleftrightarrow ({\bf C3} \lor {\bf C4})
\end{equation}
where \textbf{C1} stand for ``$S^D$ is not a Dyck word'', \textbf{C2} for ``$S^E$ is not a Dyck word'', \textbf{C3} for ``$S_o'$ is not a Dyck word'', \textbf{C4} for ``$S_v'$ is not a Dyck word''.

When evaluating the expectation value of $\hat{C}$ relatively to the Fermi vacuum, one can only encounter four possibilities:

\textit{Case 1}. 
The $\hat{a}_r^\dag\hat{a}_s$ pair corresponds to a deexcitation operator.
The Fermi-$\mathcal{L}_{1,1}$ translation of $\hat{C}$ is the same as its Fermi-$\mathcal{L}_{1}$ translation, and has been termed $S^D$ above. According to Proposition \ref{prop:nullity_F_L1}, if $S^D$ is not a Dyck word (\textbf{C1}) the expectation value of $\hat{C}$ is equal to zero:

\begin{equation}\label{eq:C1Zero}
\textbf{C1} \Longrightarrow \braket{\hat{C}}_{\Psi_0} = 0.
\end{equation}

\textit{Case 2}. 
The $\hat{a}_r^\dag\hat{a}_s$ pair corresponds to an excitation operator. As in the previous case, if the Fermi-$\mathcal{L}_{1}$ translation of $\hat{C}$, named $S^E$ above, is not a Dyck word (\textbf{C2}) the expectation value of $\hat{C}$ is equal to zero:

\begin{equation}\label{eq:C2Zero}
\textbf{C2} \Longrightarrow \braket{\hat{C}}_{\Psi_0} = 0.
\end{equation}

\textit{Case 3}. Both operators revealed to be \textit{o-}operators. If $S_o$ is a Dyck word, expulsing the two square brackets corresponding to $\hat{a}_r^\dag$ and $\hat{a}_s$ results in a Dyck word — in agreement with Proposition \ref{prop:expulsion}. The contraposition of what precedes allows to assert that if $S_o'$ is not a Dyck word, then $S_o$ is not a Dyck word too. Therefore, according to \ref{prop:criteria_nullity_EVC_fermi_2nd}, if $S_o'$ is not a Dyck word (\textbf{C3}) the expectation value of $\hat{C}$ is equal to zero:

\begin{equation}\label{eq:C3Zero}
\textbf{C3} \Longrightarrow \braket{\hat{C}}_{\Psi_0} = 0.
\end{equation}

\textit{Case 4}. If both operators revealed to be \textit{v-}operators, a similar relationship that the one found in the previous case exists. Some precautions have to be taken on the position of $\hat{a}_r^\dag\hat{a}_s$ pair in $\hat{C}$. However a more advanced study of this case is not necessary for proving our proposition: Combining \eqref{eq:C1Zero}, \eqref{eq:C2Zero}, and \eqref{eq:C3Zero} one finds
\begin{equation}\label{eq:negP1P2P3prebilan}
(\textbf{C1}\land \textbf{C2}\land (\textbf{C3}\lor \textbf{C4})) \Longrightarrow \braket{\hat{C}}_{\Psi_0} = 0.
\end{equation}
From \eqref{eq:negP1P2P3} we already know that
\begin{equation}\label{eq:negP1P2P3prebilan2}
((\neg \textit{P1})\land (\neg \textit{P2}) \land (\neg \textit{P3})) \Longleftrightarrow (\textbf{C1}\land\textbf{C2}\land(\textbf{C3}\lor \textbf{C4})).
\end{equation}
Combining \eqref{eq:negP1P2P3prebilan} with \eqref{eq:negP1P2P3prebilan2} we get
\begin{equation*} 
((\neg \textit{P1})\land (\neg \textit{P2}) \land (\neg \textit{P3})) \Longrightarrow \braket{\hat{C}}_{\Psi_0} = 0,
\end{equation*}
which is the desired result.
\end{proof}
\noindent{Our sufficient conditions for the nullity of the expectation value relatively to the Fermi vacuum of a chain containing (de)excitation operators and one arbitrary creation-annihilation operator pair using Proposition \ref{prop:nullity_F_L1_extend} is illustrated in Figure \ref{fig:nullity_F_L1_extend}.}
\begin{figure}[h!]\centering
\begin{tikzpicture}

\node (step_0) at (0,0)  {$\hat{C}$};
\node (step_1) at (3,0)  {$S$};
\node (step_2E) at (7,0)  {$S^E$};
\node (step_2D) at (7,-1)  {$S^D$};
\node (step_2') at (7,+1)  {$S'_o\& S'_v$};

\node (b2) at (8.5,0) {};

\node (step_3_D)  at (13.7,1) {$\braket{\hat{C}}_{\Psi_0} \in \{-1,0,1\}$};
\node (step_3_nD) at (13,-1)  {$\braket{\hat{C}}_{\Psi_0} = 0$};

\draw[->] (7.8,1)  to node[above] {One of them}  node[below] {is a Dyck word} (step_3_D)  ;
\draw[->] (7.8,-1) to node[above] {None of them} node[below] {is a Dyck word} (step_3_nD);

\draw[->] (step_0) to node[above] {Fermi-$\mathcal{L}_{1,1}$} (step_1);
\draw[->] (step_1) to                node[above] {``--''$\rightarrow$ ``)''}   (step_2E);
\draw[->] (step_1) to[bend right=30] node[above] {``--''$\rightarrow$ ``(''}   (step_2D);
\draw[->] (step_1) to[bend  left=30] node[above] {``--''$\rightarrow$ ``\,'' } (step_2');

\draw (7.5,1.5) -- ++(.3,0) -- ++(0,-3) -- ++(-.3,0);

\end{tikzpicture}
\caption{Sufficient conditions for the nullity of the expectation value relatively to the Fermi vacuum of a $\hat{C}$ chain containing (de)excitation operators and one arbitrary creation-annihilation operator pair. In the diagram, $S$ is written in the $\mathcal{L}_{1,1}$ alphabet. After removing the dash or replacing it by an opening or closing bracket, the resulting string — i.e., $S^E$, $S^D$, or $S_o'\&S_v'$ — is written in the $\mathcal{L}_{1}$ alphabet.}
\label{fig:nullity_F_L1_extend}
\end{figure}
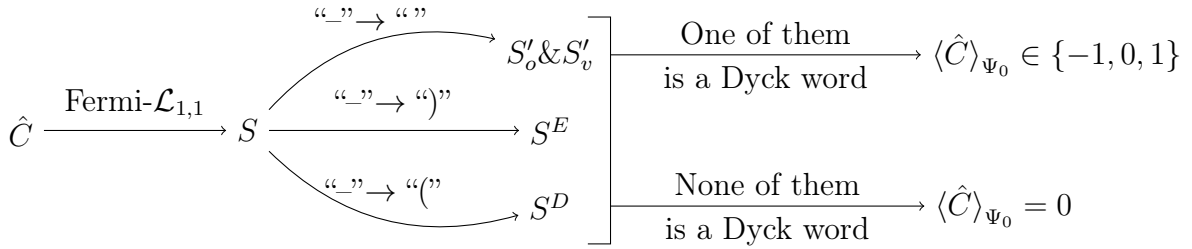

\begin{example}
Consider the three translations in {\normalfont Example \ref{ex:translation_Fermi_L-1-1}}.
The $\bm{(\,-\,)}$ string is a Dyck word if the dash is removed. Based on {Proposition {\normalfont \ref{prop:nullity_F_L1_extend}}} through its third criterion, one cannot decide whether the expectation value relatively to the Fermi vacuum of $\hat{D}_a^i\hat{a}^\dag_r\hat{a}_s\hat{E}_j^b$ is zero or is different from zero.

The second string, i.e., $\bm{(\,-\,)\,)}$ is a Dyck word if the dash is replaced by a opening bracket. Based on {Proposition {\normalfont \ref{prop:nullity_F_L1_extend}}} through its first criterion, one cannot decide whether the expectation value relatively to the Fermi vacuum of $\hat{D}_a^i\hat{a}^\dag_r\hat{a}_s\hat{E}_j^b\hat{E}_k^c$ is zero or is different from zero.

The last translation, i.e., $\bm{(\,(\,(\,-}$, cannot be turned into a Dyck word using any transformation introduced in {\normalfont Proposition \ref{prop:nullity_F_L1_extend}}. According to {\normalfont Proposition \ref{prop:nullity_F_L1_extend}}, one can assert without further investigation that the expectation value relatively to the Fermi vacuum of $\hat{D}_a^i\hat{D}_b^j\hat{a}^\dag_r\hat{a}_s$ is equal to zero.
\end{example}

\noindent In a similar way we extended the supplementation of $\mathcal{L}_1$ into $\mathcal{L}_2$ in order to account for the possibility of more than one splitting, we propose to extend the supplementation of $\mathcal{L}_1$ into $\mathcal{L}_{1,1}$ and to generalize this beyond $\mathcal{L}_1$:
\begin{equation*}
\forall i \in \llbracket 1,L\rrbracket,\, \forall k\geq 1, \, \mathcal{L}_{i,k} \coloneqq \mathcal{L}_i \cup \left\{ {\bm{-}}_k   \right\}.
\end{equation*}
where ``$\bm{-}_k$'' is a shorthand notation for a pile of $k$ dashes: for $(k=1)$ we have that ``$\bm{-}_1$'' is simply ``$\bm{-}$''; for $(k=2)$ we have that ``$\bm{-}_2$'' is ``$\bm{=}$'', etc. A pile of $k$ dashes is the translation of a chain composed, reading from the left to the right, of $k$ creation operators followed by $k$ annihilation operators, all pointing at arbitrary spin-orbital functions from $\mathcal{B}$. Therefore, one should not confuse an expression like ``$\bm{=}$'' with ``$\bm{{-}-}$''. For example, consider the $(p,q,r,s)$ sequence of integers. The ``$\bm{=}$'' expression will be the Fermi-$\mathcal{L}_{1,2}$ translation of the $\hat{a}^\dag_p\hat{a}^\dag_q\hat{a}_r\hat{a}_s$ chain while the ``${\bm{{-}-}}$'' expression will be the Fermi-$\mathcal{L}_{1,1}$ translation of the $\hat{a}^\dag_p\hat{a}_q\hat{a}_r^\dag\hat{a}_s$ chain. This extension should reveal to be useful when considering many-body operators. However, we could expect the $\mathcal{L}_{1,k}$ supplementations of $\mathcal{L}_1$ to be the most practical extension among those. 

\subsubsection{Augmented second quantization operators}

When studying second quantization relatively to any one-determinant state $\Psi$, \textit{augmented} second quantization operators can be used: To each spin-orbital $\varphi_r$ in $\mathcal{B}$ we associate an augmented creation operator, i.e.,
\begin{equation*}
{}^\Psi\hat{p}_r^\dag \coloneqq n_r^\Psi\hat{a}_r^\dag + (1-n_r^\Psi)\hat{a}_r,
\end{equation*}
\noindent and an augmented annihilation operator, i.e.,
\begin{equation*}
{}^\Psi\hat{p}_r \coloneqq (1-n_r^\Psi)\hat{a}_r^\dag +  n_r^\Psi\hat{a}_r.
\end{equation*}
\noindent Those operators are designed to behave relatively to the one-determinant state as creation and annihilation operators behave relatively to the physical vacuum. We know that, for each spin-orbital $\varphi_r$ in $\mathcal{B}$ we have, relatively to the physical vacuum, that
\begin{equation}
\bra{\;} \hat{a}_r^\dag = 0\bra{\;}, \quad \hat{a}_r \ket{\;} = \ket{0_0}.
\end{equation}
\noindent The augmented second quantization operators pointing at each spin-orbital $\varphi_r$ in $\mathcal{B}$ behave similarly but relatively to the Fermi vacuum:
\begin{equation}
\bra{\Psi_0}\hat{p}_r^\dag = 0\bra{\Psi_0}, \quad (\hat{p}_r \ket{\Psi_0} = \ket{0_{N-1}})\lor(\hat{p}_r \ket{\Psi_0} = \ket{0_{N+1}}).
\end{equation}
We dropped the ``$\Psi_0$'' superscript for the sake of readability. Any chain of \textit{o-} and \textit{v-}operators can be rewritten as a chain of augmented operators. For every $i$ integer between $1$ and $N$,
\begin{equation*}
\hat{o}_i = \hat{p}_i^\dag, \quad\hat{o}_i^\dag = \hat{p}_i,
\end{equation*}
and for every $a$ integer between $(N+1)$ and $L$,
\begin{equation*}
\hat{v}_a = \hat{p}_a, \quad\hat{v}_a^\dag = \hat{p}_a^\dag.
\end{equation*}

\begin{definition}[Augmented-$\mathcal{L}_1$ translation]\label{ex:A_L1}
Augmented-$\mathcal{L}_1$ translating a chain of augmented operators consists in replacing each augmented annihilation operator by an opening bracket {\normalfont ``}$\bm{(}${\normalfont ''}, and each augmented creation operator by a closing bracket {\normalfont ``}$\bm{)}${\normalfont ''}:
\begin{align*}
\forall r \in \llbracket 1, L \rrbracket, \quad 
&\hat{p}_r     {\longrightarrow}_{A-\mathcal{L}_1} \bm{(},\\ 
&\hat{p}^\dag_r {\longrightarrow}_{A-\mathcal{L}_1} \;\bm{)}.
\end{align*}
\end{definition}

\begin{example}\label{ex:translation_A_L1}
Let $(p,q,r,s)$ be a 4-tuple of integers, all belonging to $\llbracket 1,L\rrbracket$. Consider the $\hat{p}_p\hat{p}_q\hat{p}^\dag_s\hat{p}^\dag_r$ chain. Its augmented-$\mathcal{L}_1$ translation is:
\begin{equation*}
\hat{p}_p\hat{p}_q\hat{p}^\dag_s\hat{p}^\dag_r {\longrightarrow}_{A-\mathcal{L}_1} \bm{((\;))}.
\end{equation*}
\end{example}

\begin{proposition}\label{prop:nullity_A_L1}%[Nullity criterion in physical vacuum]
A chain of augmented operators has a zero expectation value relatively to the Fermi vacuum if its augmented-$\mathcal{L}_1$ translation results into a string of bracket(s) that is not a Dyck word.
\end{proposition}
\begin{proof} 
This can be proved by following the proof of {\normalfont Proposition \ref{prop:criteria_nullity_EVC_2nd}}.
\end{proof}

\begin{example}
Consider the chain of augmented operators defined in {\normalfont Example \ref{ex:translation_A_L1}}. Its augmented-$\mathcal{L}_1$ translation is a Dyck word. Consequently, the expectation value of that chain relatively to the Fermi vacuum can be non-zero.
\end{example}

\begin{remark}
If operators in the chain are \textit{o-} or \textit{v-}operators, this criterion is less strong than the $\mathcal{L}_2$-split translation nullity criterion. Indeed, if, relatively to the Fermi vacuum, $\varphi_p$ and $\varphi_q$ are both occupied spin orbitals and $\varphi_r$ and $\varphi_s$ are both virtual spin orbitals, the $\mathcal{L}_2$-split translation of $\hat{p}_p\hat{p}_q\hat{p}^\dag_s\hat{p}^\dag_r$ is $\bm{((}\;\&\;\bm{))}$. The resulting pair of strings is not a pair of Dyck words, hence the expectation value of $\hat{p}_p\hat{p}_q\hat{p}^\dag_s\hat{p}^\dag_r$ relatively to the Fermi vacuum is equal to zero. This conclusion cannot be obtained from the augmented-$\mathcal{L}_1$ translation alone.
\end{remark}

\subsection{Nullity criteria based on depth}
Previous sections established nullity criteria of the expectation value of a chain of second quantization operators based on the sequence of opening and closing brackets for a given translation. We insist on the fact that if the translation of the chain is a Dyck word, or a pair of Dyck words, the expectation value relatively to a given vacuum can be zero or non-zero. In this paragraph we introduce a new nullity criterion based on depths. 

\begin{proposition}\label{prop:criteria_depth_physical}
Let $\hat{C}$ be a chain of second quantization operators whose $\mathcal{L}_1$-translation is a Dyck word.
The expectation value of $\hat{C}$ relatively to the physical vacuum is equal to zero if one element in the list of depths corresponding to its $\mathcal{L}_1$-translation is strictly greater than $L$.
\end{proposition}

\begin{proof}
{Let $S$ be the $\mathcal{L}_1$-translation of $\hat{C}$. In the hypotheses of the proposition, $S$ is a Dyck word. 
Assume that the $k^\text{th}$ position in $S$ has a depth $d_k$ strictly greater than $L$. Let us consider the $M$ operators on the right of the $k^\text{th}$ operator of $\hat{C}$ read from the left to the right. We denote this subchain by $\hat{C}_M$, and apply $\hat{C}_M$ to $\ket{\;}$. 
According to definitions \ref{def:dyck_word}, \ref{def:depth}, and \ref{def:L1} among $\hat{C}_M$ there are $d_k$ more creation operators than annihilation operators.
If any of the two situations described in \eqref{eq:zero_creation} or \eqref{eq:zero_annihilation} is met, the $\braket{\hat{C}}$ expectation value is readily zero. Otherwise, we expect to find $d_k$ spin-orbitals involved in $\hat{C}_M \ket{\;}$. In that case, $d_k$ is greater than the total number of spin-orbitals in $\mathcal{B}$.
As a consequence, at a certain point in the application of $\hat{C}_M$ to $\ket{\;}$, a creation operator has been applied to an element of $\mathcal{F}_L$, leading to $\ket{0_L}$ according to \eqref{eq:zero_creation_FL}, and to a value for $\braket{\hat{C}}$ that is zero.}
\end{proof}

\begin{example}
Consider a $2$-dimensional orthonormal basis, namely $(\varphi_1, \varphi_2)$. We also consider $\hat{C}$, the $\hat{a}_1\hat{a}_2\hat{a}_1\hat{a}_1^\dag\hat{a}_1^\dag\hat{a}_2^\dag$ chain of second quantization operators defined in that basis. The $\mathcal{L}_1$-translation of the $\hat{C}$ is $\bm{(\;(\;(\;)\;)\;)}$. It is a Dyck word, whose depth sequence is $(0,1,2,3,2,1,0)$. The maximal depth reached in the translation is equal to 3, which is strictly greater than the size of the basis. Therefore, according to {\normalfont Proposition \ref{prop:criteria_depth_physical}} the expectation value of $\hat{C}$ relatively to the physical vacuum is equal to zero.
\end{example}

\noindent A similar nullity criterion exists for expectation values relatively to the Fermi vacuum, as reported in Proposition \ref{prop:criteria_depth_fermi}.

\begin{proposition}\label{prop:criteria_depth_fermi}
Let $\hat{C}$ be a chain of second quantization operators whose $\mathcal{L}_2$-split translation $S_o\&S_v$ is a pair of Dyck words.
The expectation value of $\hat{C}$ relatively to the Fermi vacuum is equal to zero if {\normalfont{(i)}} a position in $S_o$ has a depth strictly greater than $N$ or if {\normalfont{(ii)}} a position in $S_v$ has a depth strictly greater than $(L-N)$.
\end{proposition}
\begin{proof}

Let $M_o$ (respectively, $M_v$) denote the number of \textit{o}-operators (respectively, \textit{v}-operators) in $\hat{C}$. The operators in $\hat{C}$ can be reordered in two ways in order to provide two new chains, namely $\hat{C}_o$ and $\hat{C}_v$. The former is structured as follows: reading it from the right to the left, all the \textit{o}-operators are first met in the same order as in $\hat{C}$; all the \textit{v-}operators are then met, in the same order as in $\hat{C}$. The latter, i.e., $\hat{C}_v$, is constructed in a similar fashion: all the \textit{v}-operators are first met in the same order as in $\hat{C}$, before the \textit{o}-operators are met, in the same order as in $\hat{C}$. The only effect of such a reorganization of the operators can be a change in the sign of the expectation value. Therefore, if the expectation value of $\hat{C}_o$ (respectively, $\hat{C}_v$) relatively to the Fermi vacuum is equal to zero, so is $\braket{\hat{C}}_{\Psi_0}$.

(i) Assume that the $k_o^\text{th}$ position in $S_o$ has a depth $d_{k_o}$ strictly greater than $N$. Let us consider the $m_o$ operators on the right of the ${(M_v+k_o)}^\text{th}$ operator of $\hat{C}_o$ read from the left to the right. We denote this subchain by $\hat{C}_{m_o}$, and apply $\hat{C}_{m_o}$ to $\ket{\Psi_0}$. 
According to definitions \ref{def:dyck_word}, \ref{def:depth}, and \ref{def:L2} among $\hat{C}_{m_o}$ there are $d_{k_o}$ more \textit{o-}annihilation operators than \textit{o-}creation operators.
If any of the two situations described in \eqref{eq:zero_creation} or \eqref{eq:zero_annihilation} is met, the $\braket{\hat{C}_o}_{\Psi_0}$ expectation value is zero. Otherwise, we expect to find $(N-d_{k_o})$ occupied spin-orbitals involved in $\hat{C}_{m_o} \ket{\Psi_0}$. However, $(N-d_{k_o})$ is strictly negative.
This implies that at a certain point in the application of $\hat{C}_{m_o}$ to $\ket{\Psi_0}$ an \textit{o-}annihilation has been applied to the physical vacuum state, leading to $\hat{C}_{m_o}\ket{\Psi_0}$ being equal to $\ket{0_0}$. Therefore, the $\braket{\hat{C}_o}_{\Psi_0}$ expectation value is equal to zero. Hence, $\braket{\hat{C}}_{\Psi_0}$ is equal to zero.

(ii) Assume that the $k_v^\text{th}$ position in $S_v$ has a depth $d_{k_v}$ strictly greater than $(L-N)$. Let us consider the $m_v$ operators on the right of the $(M_o+{k_v})^\text{th}$ operator of $\hat{C}_v$ read from the left to the right. We denote this subchain by $\hat{C}_{m_v}$, and apply $\hat{C}_{m_v}$ to $\ket{\Psi_0}$. 
According to definitions \ref{def:dyck_word}, \ref{def:depth}, and \ref{def:L2} among $\hat{C}_{m_v}$ there are $d_{k_v}$ more \textit{v-}creation operators than \textit{v-}annihilation operators.
If any of the two situations described in \eqref{eq:zero_creation} or \eqref{eq:zero_annihilation} is met, the $\braket{\hat{C}_v}_{\Psi_0}$ expectation value is zero. Otherwise, we expect to find $d_{k_v}$ virtual spin-orbitals involved in $\hat{C}_{m_v} \ket{\Psi_0}$. In that case, $d_{k_v}$ is greater than the total number of virtual spin-orbitals — i.e., $(L-N)$.
This implies that at a certain point in the application of $\hat{C}_{m_v}$ to $\ket{\Psi_0}$ a \textit{v-}creation has been applied to a determinant state already containing all virtual spin-orbitals, leading to $\hat{C}_{m_v}\ket{\Psi_0}$ being equal to $\ket{0_{L}}$. Therefore, the $\braket{\hat{C}_v}_{\Psi_0}$ expectation value is equal to zero. Hence, $\braket{\hat{C}}_{\Psi_0}$ is equal to zero.
\end{proof}

\begin{example}
Consider the case where $(L=3)$ and $(N=2)$. We have a $3$-dimensional orthonormal basis, namely $(\varphi_1, \varphi_2, \varphi_3)$, and our Fermi vacuum, $\Psi_0$, is simply $(\varphi_1 \wedge \varphi_2)$. We also consider $\hat{C}$, the 
\begin{equation*}
\hat{o}_1^\dag \hat{v}_3 \hat{v}_3\hat{v}_3^\dag\hat{v}_3^\dag\hat{o}_1
\end{equation*}
chain of second quantization operators defined in that basis. The $\mathcal{L}_2$-translation of the $\hat{C}$ is $\bm{(\;)}\&\bm{(\;(\;)\;)}$. It is a pair of Dyck words. The depth sequence for virtual operators is $(0,1,2,1,0)$. The depth sequence for occupied operators is $(0,1,0)$.  

For the occupied spin-orbitals, the nullity criterion is not met, but for virtual spin-orbitals, the maximal depth reached in the translation is equal to 2. This maximal depth is striclty greater than $(L-N)$. 

Therefore, according to {\normalfont Proposition \ref{prop:criteria_depth_fermi}} the expectation value of $\hat{C}$ relatively to the Fermi vacuum is equal to zero.
\end{example}

\begin{remark} These criteria can be extended for any partition of $\mathcal{B}$: Let the $\mathcal{B}$ ordered basis be partitioned into $i$ parts, namely $(\mathcal{B}_1,\ldots,\mathcal{B}_i)$. For every $k$ in $\llbracket 1, i\rrbracket$, the $k^\text{th}$ part contains $n_k$ spin-orbitals. Picking up a number, $\ell$, between zero and $i$, we can define ``$\ell$--reference'' states relatively to the partition as the states in which all the spin-orbitals of exactly $\ell$ parts of $(\mathcal{B}_1,\ldots,\mathcal{B}_i)$ are present exactly once.

For a given $\ell$--reference, consider a chain $\hat{C}$ of second quantization operators and its $\mathcal{L}_i$-split translation $S_1\&\cdots\&S_i$. If the $k^\text{th}$ word in $S_1\&\cdots\&S_i$ reaches a depth greater than $n_k$, the expectation value of $\hat{C}$ relatively to the chosen $\ell$--reference is equal to zero. 
\end{remark}

\subsection{Summary}

\noindent In this paper we have introduced multiple alphabets and five main translation rules for rewriting chains of second quantization operators. Table \ref{tab:summarytranslations} summarizes these constructions and properties, and shows them in action with chosen examples.
\begin{table}[h!]
\def\arraystretch{1.5}
\begin{center}
\begin{tabular}{ccccc}
\hline 
Name & Alphabet & Definition & Proposition & Example \\ 
\hline 
$\mathcal{L}_1$             & \{ $\bm{)}$, $\bm{(}$ \} & \ref{def:L1} & \ref{prop:nullity_L1} & $\hat{a}_p\hat{a}_q^\dag\hat{a}_r\hat{a}_s^\dag \longrightarrow_{\mathcal{L}_1} \bm{(\,)\,(\,)}$ \\ 
$\mathcal{L}_{2,\mathcal{S}}$ & \{ $\bm{)}$, $\bm{(}$, $\bm{]}$, $\bm{[}$ \} & \ref{def:L2} & \ref{prop:nullity_translation_Fermi} & $\hat{o}_i^\dag\hat{v}_a\hat{v}_b^\dag\hat{o}_j \longrightarrow_{\mathcal{L}_{2,\mathcal{S}}} \bm{[\,]\,\&\,(\,)}$ \\ 
Fermi-$\mathcal{L}_{1}$     & \{ $\bm{)}$, $\bm{(}$ \} & \ref{def:F_L1} & \ref{prop:nullity_F_L1} & $\hat{o}_i^\dag\hat{v}_a\hat{v}_b^\dag\hat{o}_j = \hat{D}_a^i\hat{E}_j^b\stackrel{\mathcal{F}}{\longrightarrow}_{\mathcal{L}_{1}} \bm{(\,)}$\\ 
Fermi-$\mathcal{L}_{1,1}$     & \{ $\bm{)}$, $\bm{(}$, $\bm{-}$ \} & \ref{def:F_L1-extend} & \ref{prop:nullity_F_L1_extend} & $\hat{o}_i^\dag\hat{v}_a\hat{a}_r^\dag\hat{a}_s\hat{v}_b^\dag\hat{o}_j = \hat{D}_a^i\hat{a}_r^\dag\hat{a}_s\hat{E}_j^b\stackrel{\mathcal{F}}{\longrightarrow}_{\mathcal{L}_{1,1}} \bm{(\,-\,)}$ \\ 
Augmented-$\mathcal{L}_1$ & \{ $\bm{)}$, $\bm{(}$ \} & \ref{ex:A_L1} & \ref{prop:nullity_A_L1} & $\hat{p}_p\hat{p}_q\hat{p}^\dag_s\hat{p}^\dag_r \longrightarrow_{A-\mathcal{L}_{1}} \bm{(\,(\,)\,)}$\\
\hline 
\end{tabular} 
\caption{The five main translations introduced in this paper with the corresponding alphabet, together with a reference to their definition, and a reference to the proposition where nullity criteria are stated using the translation of interest, and an illustrative example.}
\label{tab:summarytranslations}
\end{center}
\end{table}
\section{Perspectives}

\subsection{Dyck paths}
\begin{definition}[Dyck path]\label{def:dyck_path}
A Dyck path of length $2n$ is a lattice path, starting at the origin $(0,0)$ and ending at $(2n,0)$, which never drops below the $x$-axis (materialized by a dashed horizontal line below). Each step is either an upward step $(1,1)$ or a downward step $(1,-1)$.
\end{definition}

\begin{example}\label{ex:DyckPaths_from_DyckWords}
All Dyck paths of length 6, \textit{i.e.}, using 3 upward steps and 3 downward steps, are:
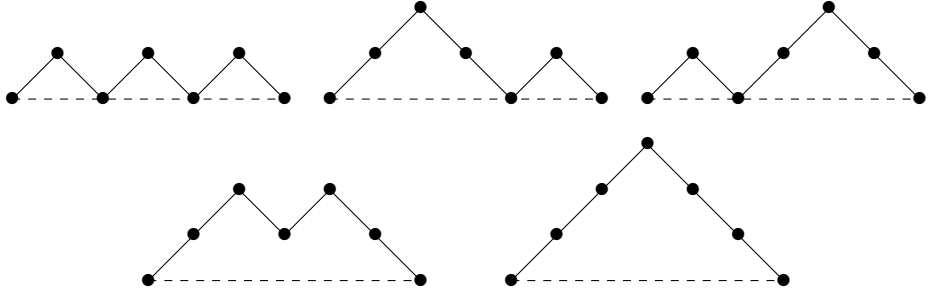
\begin{center}\begin{normalfont}
\begin{tikzpicture}[scale=0.6]
\draw[dashed] (0,0)--(6,0);
\draw[black] (0,0)node{•}-- ++(1,1)node{•}-- ++(1,-1)node{•}-- ++(1,1)node{•}-- ++(1,-1)node{•}-- ++(1,1)node{•}-- ++(1,-1)node{•};

\draw[dashed] (7,0)--(13,0);
\draw[black] (7,0)node{•}-- ++(1,1)node{•}-- ++(1,1)node{•}-- ++(1,-1)node{•}-- ++(1,-1)node{•}-- ++(1,1)node{•}-- ++(1,-1)node{•};

\draw[dashed] (14,0)--(20,0);
\draw[black] (14,0)node{•}-- ++(1,1)node{•}-- ++(1,-1)node{•}-- ++(1,1)node{•}-- ++(1,1)node{•}-- ++(1,-1)node{•}-- ++(1,-1)node{•} ;

\draw[dashed] (3,-4)--(9,-4);
\draw[black] (3,-4)node{•}-- ++(1,1)node{•}-- ++(1,1)node{•}-- ++(1,-1)node{•}-- ++(1,1)node{•}-- ++(1,-1)node{•}-- ++(1,-1)node{•} ;

\draw[dashed] (11,-4)--(17,-4);
\draw[black] (11,-4)node{•}-- ++(1,1)node{•}-- ++(1,1)node{•}-- ++(1,1)node{•}-- ++(1,-1)node{•}-- ++(1,-1)node{•}-- ++(1,-1)node{•} ;
\end{tikzpicture}
\end{normalfont}\end{center}
\end{example}

\noindent To every Dyck path corresponds a unique Dyck word. Indeed, during a left-to-right reading of a Dyck word, if to each opening (respectively, closing) bracket in a Dyck word is associated an upward step (respectively, downward step) the resulting path is a Dyck path. There is a one-to-one correspondence between the five Dyck paths in Example \ref{ex:DyckPaths_from_DyckWords} and the five Dyck words in Example \ref{ex:dyck_word}.

A direct consequence of this bijection is that for every translation from a chain of second quantization operators to a string of brackets there is an equivalent translation to a path. 
From each translation into path, nullity criteria of the expectation value can be found depending on the reference state or the type — occupied, virtual, arbitrary — of operator that is implied. As a consequence, each translation reported in Table \ref{tab:summarytranslations} has its equivalent in the Dyck path diagrammatic language. For the sake of brevity, solely the equivalent of the Fermi-$\mathcal{L}_1$ translation will be explicitly developed below.
\begin{definition}[Path Fermi-$\mathcal{L}_1$ translation]\label{def:path_F_L1}
Path Fermi-$\mathcal{L}_1$ translating a chain of (de)excitation operators consists in initializing the path in $(0,0)$ and, reading the chain from the left to right adding an upward (respectively, downard) step to the path at each deexcitation (respectively, excitation) operator.
\end{definition}

\begin{proposition}\label{prop:path_nullity_F_L1}
The expectation value of a chain of (de)excitation operators relatively to the Fermi vacuum is zero if its path Fermi-$\mathcal{L}_1$ translation is not a Dyck path.
\end{proposition}

\begin{proof}
If the path Fermi-$\mathcal{L}_1$ translation of a chain is not a Dyck path, its Fermi-$\mathcal{L}_1$ translation is not a Dyck word. Therefore, in virtue of Proposition \ref{prop:nullity_F_L1}, its expectation value relatively to the Fermi vacuum is equal to zero.
\end{proof}

\begin{example}\label{ex:FermiDyckPath}
Let $(i_k)_{1\leq k\leq 6}$ be a 6-tuple of integers, all belonging to $\llbracket 1,N\rrbracket$. 
Let $(a_k)_{1\leq k\leq 6}$ be a 6-tuple of integers, all belonging to $\llbracket N+1,L\rrbracket$. Consider the $\hat{C}_1 := \hat{D}_{a_1}^{i_1}\hat{D}_{a_2}^{i_2}\hat{E}_{i_3}^{a_3}\hat{D}_{a_4}^{i_4}\hat{E}_{i_5}^{a_5}\hat{E}_{i_6}^{a_6}$ and $\hat{C}_2 := \hat{D}_{a_1}^{i_1}\hat{E}_{i_2}^{a_2}\hat{E}_{i_3}^{a_3}\hat{D}_{a_4}^{i_4}\hat{D}_{a_5}^{i_5}\hat{E}_{i_6}^{a_6}$ chains. Their respective Fermi-$\mathcal{L}_1$ translation and their path Fermi-$\mathcal{L}_1$ translation are illustrated in {\normalfont Figure \ref{fig:dyck_path_representation}}.

\begin{figure}[h!]\centering 
\begin{normalfont}
\begin{tikzpicture}[scale=0.8]
\draw[dashed] (0,0)--(-6,0);
\draw(0,0)node{•}-- ++(-1,1)node{•} -- ++(-1,1)node{•} -- ++(-1,-1)node{•} -- ++(-1,1)node{•} -- ++(-1,-1)node{•}-- ++(-1,-1)node{•};

\def\x{-.5}
\draw (-5.5,\x) node{$\hat{D}_{a_1}^{i_1}$}; \draw (-4.5,\x) node{$\hat{D}_{a_2}^{i_2}$}; 
\draw (-3.5,\x) node{$\hat{E}_{i_3}^{a_3}$}; \draw (-2.5,\x) node{$\hat{D}_{a_4}^{i_4}$}; 
\draw (-1.5,\x) node{$\hat{E}_{i_5}^{a_5}$}; \draw (-0.5,\x) node{$\hat{E}_{i_6}^{a_6}$};

\def\x{-1.5}
\draw (-5.5,\x) node{$\bm{(}$}; \draw (-4.5,\x) node{$\bm{(}$}; 
\draw (-3.5,\x) node{$\bm{)}$}; \draw (-2.5,\x) node{$\bm{(}$}; 
\draw (-1.5,\x) node{$\bm{)}$}; \draw (-0.5,\x) node{$\bm{)}$};

\tikzset{shift={(8,0)}}

\draw[dashed] (0,0)--(-6,0);
\draw(0,0)node{•}-- ++(-1,1)node{•} -- ++(-1,-1)node{•} -- ++(-1,-1)node{•} -- ++(-1,1)node{•} -- ++(-1,1)node{•}-- ++(-1,-1)node{•};

\def\x{-1.5}
\draw (-5.5,\x) node{$\hat{D}_{a_1}^{i_1}$}; \draw (-4.5,\x) node{$\hat{E}_{i_2}^{a_2}$}; 
\draw (-3.5,\x) node{$\hat{E}_{i_3}^{a_3}$}; \draw (-2.5,\x) node{$\hat{D}_{a_4}^{i_4}$}; 
\draw (-1.5,\x) node{$\hat{D}_{a_5}^{i_5}$}; \draw (-0.5,\x) node{$\hat{E}_{i_6}^{a_6}$};

\def\x{-2.5}
\draw (-5.5,\x) node{$\bm{(}$}; \draw (-4.5,\x) node{$\bm{)}$}; 
\draw (-3.5,\x) node{$\bm{)}$}; \draw (-2.5,\x) node{$\bm{(}$}; 
\draw (-1.5,\x) node{$\bm{(}$}; \draw (-0.5,\x) node{$\bm{)}$};
\end{tikzpicture}
\end{normalfont}
\caption{(Left) Path Fermi-$\mathcal{L}_1$ translation and Fermi-$\mathcal{L}_1$ translation of $\hat{C}_1$ from Example \ref{ex:FermiDyckPath}. (Right) path Fermi-$\mathcal{L}_1$ translation and Fermi-$\mathcal{L}_1$ translation of $\hat{C}_2$ from Example \ref{ex:FermiDyckPath}. For the sake of clarity, each operator is placed below its corresponding step.}
\label{fig:dyck_path_representation}
\end{figure}
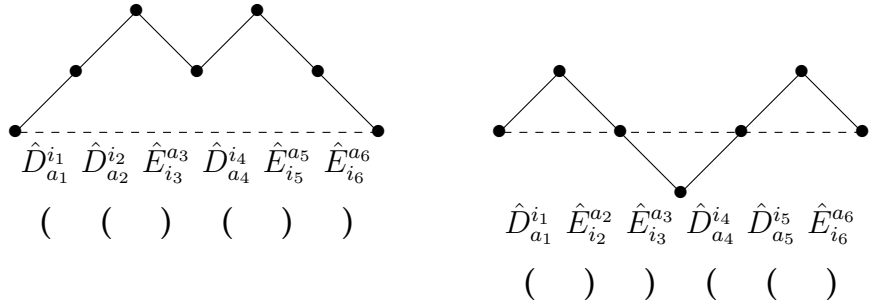

\noindent The path corresponding to $\hat{C}_2$ crosses the x-axis, its expectation value relatively to the Fermi vacuum is equal to zero. On the other hand, the path corresponding to $\hat{C}_1$ is a Dyck path. Therefore, its expectation value relatively to the Fermi vacuum can be non-zero.
\end{example}

\noindent More than just adding a new possible graphical representation of chains of second quantization operators, the path representation is more adapted to chains containing an odd number of operators.

\begin{proposition}\label{prop:dyck_path_nullity}
Consider a chain, $\hat{C}$, of (de)excitation operators. Let $P$ be its path Fermi-$\mathcal{L}_1$ translation. Let $h$ be the height of the last node.
The transition value $\bra{\Psi_0}\hat{C}\ket{\Psi_h}$, where $\Psi_h$ is an $h$--tuply excited determinant relatively to the Fermi vacuum, is zero if any node in $P$ drops below the $x$-axis or if the height of the final node is different from $h$.
\end{proposition}

\begin{proof}
$\Psi_h$ being an $h$--tuply excited determinant relatively to the Fermi vacuum, there exists an $h$-tuple $(i_k)_{1\leq k\leq h}$ of integers, all belonging to $\llbracket 1,N\rrbracket$, and an $h$-tuple $(a_k)_{1\leq k\leq h}$ of integers, all belonging to $\llbracket N+1,L\rrbracket$, such that  $\ket{\Psi_h} =\hat{E}_{i_1}^{a_1}\cdots\hat{E}_{i_h}^{a_h}  \ket{\Psi_0}$. The $\bra{\Psi_0}\hat{C}\ket{\Psi_h}$ transition value is equal to the $\bra{\Psi_0}\hat{C}\hat{E}_{i_1}^{a_1}\cdots\hat{E}_{i_h}^{a_h}\ket{\Psi_0}$ expectation value.
The path Fermi-$\mathcal{L}_1$ translation of $\hat{C}\hat{E}_{i_1}^{a_1}\cdots\hat{E}_{i_h}^{a_h}$ starts and finishes on the $x$-axis. If a node drops below the $x$-axis, the path Fermi-$\mathcal{L}_1$ translation is not a Dyck path — in which case, according to {Proposition \ref{prop:path_nullity_F_L1}} the $\bra{\Psi_0}\hat{C}\hat{E}_{i_1}^{a_1}\cdots\hat{E}_{i_h}^{a_h}\ket{\Psi_0}$ expectation value (hence, the $\bra{\Psi_0}\hat{C}\ket{\Psi_h}$ transition value) is equal to zero.
\end{proof}

\begin{example}
Let $(i,j,k)$ be a 3-tuple of integers, all belonging to $\llbracket 1,N\rrbracket$. 
Let $(a,b,c)$ be a 3-tuple of integers, all belonging to $\llbracket N+1,L\rrbracket$. Consider the $\hat{D}_{a}^{i}\hat{D}_{b}^{j}\hat{E}_{k}^{c}$ chain. Its path Fermi-$\mathcal{L}_1$ translation and Fermi-$\mathcal{L}_1$ translation are illustrated in {\normalfont Figure \ref{fig:dyck_path_representation_odd}}. The path Fermi-$\mathcal{L}_1$ translation of the chain is not a Dyck path. Indeed, the final node of the path reaches a height of 1. According to {\normalfont Proposition \ref{prop:dyck_path_nullity}}, the $\bra{\Psi_0}\hat{D}_{a}^{i}\hat{D}_{b}^{j}\hat{E}_{k}^{c}\ket{\Psi_1}$ transition value, with $\Psi_1$ being a one-determinant state, singly excited relatively to the Fermi vacuum, can be non-zero.
\begin{figure}[h!]\centering 
\begin{normalfont}
\begin{tikzpicture}[scale=0.8]
\draw[dashed] (0,0)--(3,0);
\draw(0,0)node{•}-- ++(1,1)node{•} -- ++(1,1)node{•}-- ++(1,-1)node{•};

\def\x{-.5}
\draw (0.5,\x) node{$\hat{D}_{a}^{i}$}; 
\draw (1.5,\x) node{$\hat{D}_{b}^{j}$}; 
\draw (2.5,\x) node{$\hat{E}_{k}^{c}$};

\end{tikzpicture}
\end{normalfont}
\caption{Path Fermi-$\mathcal{L}_1$ translation and Fermi-$\mathcal{L}_1$ translation of $\hat{D}_{a}^{i}\hat{E}_{j}^{b}\hat{E}_{k}^{c}$.}
\label{fig:dyck_path_representation_odd}
\end{figure}
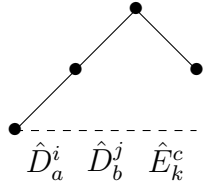
\end{example}

\noindent We see from what precedes that one advantage of this representation, is that the number of missing (de)excitation(s) can be directly read on the translation.

\subsection{Bosonic second quantization}
Until now, only the fermionic second quantization has been considered. However, similar translation can be defined for bosonic second quantization. In this subsection we consider an $L$--dimensional orthonormal ordered basis, $\mathcal{B}_b = (\phi_\alpha)_{1\leq \alpha\leq L}$. Bosonic creation and annihilation second quantization operators defined relatively to $\mathcal{B}_b$ are labelled $\hat{b}_\alpha^\dag$ and $\hat{b}_\alpha$ respectively, with $\alpha$ belonging to $\llbracket 1,L\rrbracket$.

\begin{definition}[Bosonic $\mathcal{L}_1$ translation]\label{def:L1_boson} 
Bosonic $\mathcal{L}_1$-translating a chain of individual bosonic second quantization operators consists in replacing each annihilation operator by an opening bracket {\normalfont ``}$\bm{(}${\normalfont ''}, and each creation operator by a closing bracket {\normalfont ``}$\bm{)}${\normalfont ''}:
\begin{align*}
\forall \alpha \in \llbracket 1, L \rrbracket, \quad &\hat{b}_\alpha \stackrel{b\quad}{\longrightarrow}_{\mathcal{L}_1} \bm{(},\\ &\hat{b}^\dag_\alpha \stackrel{b\quad}{\longrightarrow}_{\mathcal{L}_1} \;\bm{)}.
\end{align*}
\end{definition}

\begin{proposition}\label{prop:nullity_L1_boson}
A chain of bosonic second quantization operators has a zero expectation value relatively to the physical vacuum if its bosonic $\mathcal{L}_1$--translation results into a string of bracket(s) that is not a Dyck word.
\end{proposition}

\begin{proof}
Following the same methods as in the proof of Proposition \ref{prop:criteria_nullity_EVC_2nd}, we can prove that there is a one-to-one correspondence between the criteria defining Dyck words and criteria establishing the nullity of expectation values of bosonic second quantization operators chains relatively to the physical vacuum. 
\end{proof}

\begin{remark}
Contrary to the fermionic second quantization, occupation numbers in bosonic second quantization are not bound to zero and one. A direct consequence of that fact is the absence of nullity criteria based on depth, relatively to the physical vacuum, that would be related to an upper bound of individual states occupation number.
\end{remark}

\section{Conclusion}
We introduced and proved that there is a one-to-one correspondence between the criteria defining Dyck words and sufficient conditions establishing the nullity of expectation values of fermionic second quantization operators chains relatively to the physical vacuum or relatively to a one-determinant state. We formulated different translations of second quantization operators into bracket alphabets. For a given translation, sufficient conditions for the nullity of an expectation value of a chain are provided as simple and easy-to-check criteria on its translation. Depending on the reference state and/or the {type} of operators present in the chain, we specified which translation is relevant for revealing sufficient conditions for the nullity of the expectation value of the chain. Moreover, we introduced in second quantization formalism the depth number and the expulsion transformation.

\section*{Acknowledgements}

We would like to warmly thank Emmanuel Giner, Julien Toulouse, Saad Yalouz, Anthony Scemama, Pierre-François Loos, and colleagues from the LPCT for very fruitful discussions on the topic.

\bibliographystyle{unsrt}

\end{document}